\documentclass[11pt,a4paper, english]{article}
\usepackage[utf8]{inputenc}			% Allows the use of UTF-8 characters
\usepackage{makeidx}
\usepackage[pdftex]{graphicx}
\usepackage{mathtools,amsfonts}
\usepackage{array}
\usepackage{multirow, hhline}
\usepackage{paralist}

\usepackage[usenames,dvipsnames,table]{xcolor}
\usepackage{wrapfig}
\usepackage[english]{babel}		% Allows hiphenization of words and other stuff
\usepackage[margin=1in]{geometry}	% Defines margins to 1 inch
\usepackage{lipsum}			% provides filler text \lipsum[n]
\usepackage{comment}
\usepackage{chngpage}			% Allows width to be changed dynamically
\usepackage{amsfonts,amssymb,amsmath, amsthm}
\usepackage{amsopn}
\usepackage{color}

\usepackage{pgfplots}
\usepackage{pgfgantt}
\usepackage{tikz}
\usetikzlibrary{arrows}
\usetikzlibrary{calc}
\usepackage{algpseudocode}
\usepackage[Algorithm]{algorithm}
\usepackage{subcaption}
\usepackage{tabto}
\usepackage{csquotes} %quotations
\usepackage{multirow}
\usepackage{amsthm}

\usepackage[colorlinks,urlcolor=black,citecolor=black,linkcolor=black,menucolor=black]{hyperref}

\newtheorem{theorem}{Theorem}
\newtheorem*{theorem*}{Theorem}

\newtheorem{corollary}[theorem]{Corollary}

\newtheorem{remark}[theorem]{Remark}
\newtheorem{lemma}[theorem]{Lemma}

\newtheorem{fact}[theorem]{Fact}

\newtheorem{observation}[theorem]{Observation}
\newtheorem{claim}[theorem]{Claim}
\newtheorem{conjecture}[theorem]{Conjecture}

\newcommand{\abs}[1]{| #1 |}

\newcommand{\sset}[1]{\left\{ #1 \right\} }

\newcommand{\lref}[1]{(\ref{#1})}

\newcommand{\myBox}[3]{\begin{array}{|c|} \hline #1 \\ \hline #2 \\ \hline  \multicolumn{1}{c}{} \\[-0.7em] \multicolumn{1}{c}{(#3)} \end{array}}
\newcommand{\myBoxtwo}[2]{\begin{array}{|c|} \hline \\ {\raisebox{0.5em}[0cm][0cm]{\ensuremath{#1}}} \\ \hline \multicolumn{1}{c}{} \\[-0.7em]\multicolumn{1}{c}{(#2)} \end{array}}

\newcommand{\Allocation}[6]{\myBox{#1}{#2}{1}\hspace{2em}\myBox{#3}{#4}{2}\hspace{2em}\myBox{#5}{#6}{3}}
\newcommand{\Allocationfive}[5]{\myBox{#1}{#2}{1}\hspace{2em}\myBox{#3}{#4}{2}\hspace{2em}\myBoxtwo{#5}{3}}
\newcommand{\Allocationfivetwo}[5]{\myBox{#1}{#2}{1}\hspace{2em}\myBoxtwo{#3}{2}\hspace{2em}\myBox{#4}{#5}{3}}
\newcommand{\Allocationfiveone}[5]{\myBoxtwo{#1}{1}\hspace{2em}\myBox{#2}{#3}{2}\hspace{2em}\myBox{#4}{#5}{3}}

\newcommand{\VAL}{\mathit{VAL}_3}

\title{EFX Exists for Three Agents}
\author{Bhaskar Ray Chaudhury\thanks{MPI for Informatics, Saarland Informatics Campus, Graduate School of Computer Science, Saarbr\"ucken, Germany} \and Jugal Garg\thanks{University of Illinois at Urbana-Champaign} \and Kurt Mehlhorn\thanks{MPI for Informatics, Saarland Informatics Campus, Germany}}
\author{Bhaskar Ray Chaudhury\thanks{MPI for Informatics, Saarland Informatics Campus, Graduate School of Computer Science, Saarbr\"ucken, Germany}\\ \texttt{\small braycha@mpi-inf.mpg.de} \and Jugal Garg\thanks{University of Illinois at Urbana-Champaign. Supported by NSF Grants CCF-1755619 (CRII) and CCF-1942321 (CAREER)}\\ \texttt{\small jugal@illinois.edu}  \and Kurt Mehlhorn\thanks{MPI for Informatics, Saarland Informatics Campus, Germany}\\ \texttt{\small mehlhorn@mpi-inf.mpg.de}}
\date{}

\newcommand{\NSW}{{\mathit{NSW}}}

\begin{document}
\maketitle

\begin{abstract} 
We study the problem of distributing a set of indivisible items among agents with additive valuations in a \emph{fair} manner. The fairness notion under consideration is Envy-freeness up to \emph{any} item (EFX). %is arguably the most compelling fairness concept for this problem.
Despite significant efforts by many researchers for several years, the existence of EFX allocations has not been settled beyond the simple case of two agents. In this paper, we show constructively that an EFX allocation always exists for three agents. Furthermore, we falsify the conjecture by Caragiannis et al.~\cite{CaragiannisGravin19} by showing an instance with three agents for which there is a partial EFX allocation (some items are not allocated) with higher Nash welfare than that of any complete EFX allocation. 

%Given a set $M$ of a goods and a set of $n$ agents, an allocation is a partition of $M$ into $n$ bundles, one for each agent. The agents are assumed to have additive valuations for bundles of goods, i.e., agent $i$ has a function $v_i: M \rightarrow \R_{\ge 0}$ and $v_i(S) = \sum_{g \in S} v_i(g)$ for any set $S$ of goods. An allocation $(X_1,\ldots,X_n)$ is envy-free up to any good (EFX) if for any agents $i$ and $j$ and any good $g \in X_j$, we have $v_i(X_j \setminus g) \le v_i(X_i)$. We show: \emph{for three agents and additive valuations, an EFX allocation always exists.} An allocation $(Y_1,\ldots,Y_n)$ Pareto dominates an allocation $(X_1,\ldots,X_n)$ if $v_i(Y_i) \ge v_i(X_i)$ for all $i$ and $v_i(Y_i) > v_i(X_i)$ for at least one $i$. We show: \emph{there is an instance for three agents with additive valuations for which there is a partial EFX allocation (= only a subset of the goods is allocated) that is not Pareto dominated by any complete EFX allocation (= all goods are allocated).}
  \end{abstract}

\section{Introduction}
Discrete fair division of resources is a fundamental problem in many multi-agent settings. Here, the goal is to distribute a set $M$ of $m$ \emph{indivisible} items among $n$ agents in a \emph{fair} manner. Each agent $i$ has a valuation function $v_i \colon 2^{M} \rightarrow \mathbb{R}_{\geq 0}$ that quantifies the amount of utility agent $i$ derives from each subset of items. In case of \emph{additive} valuation functions, $v_i(S) := \sum_{j\in S} v_i(\{j\}),\ \forall S\subseteq M$. Let $X=\langle X_1, X_2, \dots, X_n\rangle$ denote a partition of $M$ into $n$ bundles such that $X_i$ is allocated to agent $i$. Among various choices, \emph{envy-freeness} is the most natural fairness concept, where no agent $i$ envies another agent $j$'s bundle, i.e., for all agents $i$, $j$ with $i \neq j$ we have $v_i(X_i) \ge v_i(X_j)$. However, an envy-free allocation does not always exist, e.g., consider allocating a single valuable item among $n\ge 2$ agents. This necessitates the study of relaxed notions of envy-freeness:

\paragraph{Envy-freeness up to \emph{one} item (EF1):} This relaxation was introduced by Budish~\cite{budish2011combinatorial}. An allocation $X$ is said to be EF1 if no agent $i$ envies another agent $j$ after the removal of \emph{some} item in $j$'s bundle, i.e., $v_i(X_i) \geq v_i(X_j \setminus g)$ for \emph{some} $g \in X_j$. So we allow $i$ to envy $j$, but the envy must disappear after the removal of some valuable item (according to agent $i$) from $j$'s bundle. Note that there is no actual removal: This is simply to assess how agent $i$ values his own bundle when compared to $j$'s bundle. It is well known that an EF1 allocation always exists, and it can be obtained in polynomial time using the famous envy-cycles procedure by Lipton et al.~\cite{LiptonMMS04}. However, an EF1 allocation may be unsatisfactory: Intuitively, EF1 insists that envy disappears after the removal of the \emph{most valuable item} according to the envying agent from the envied agent's bundle---however, in many cases, the most valuable item might be the primary reason for very large envy to exist in the first place. Therefore, stronger notions of fairness are desirable in many circumstances.

\paragraph{Envy-freeness up to \emph{any} item (EFX):} This relaxation was introduced by Caragiannis et al.~\cite{CaragiannisKMP016}. An allocation $X$ is said to be EFX if no agent $i$ envies another agent $j$ after the removal of \emph{any} item in $j$'s bundle, i.e., $v_i(X_i) \geq v_i(X_j \setminus g)$ for \emph{all} $g \in X_j$. Unlike EF1, in an EFX allocation, the envy between any pair of agents disappears after the removal of the \emph{least valuable item} (according to agent $i$) from $j$'s bundle. Note that every EFX allocation is an EF1 allocation, but not \textit{vice-versa}. Consider a simple example of two agents with additive valuations and three items $\{a, b, c\}$ from~\cite{CKMS20}, where the agents valuation for individual items are as follows, 
\begin{center}
\begin{minipage}[b]{0.3\linewidth}
\centering
\begin{eqnarray*}
  \setlength{\arraycolsep}{0.5ex}\setlength{\extrarowheight}{0.25ex}
\begin{array}{@{\hspace{1ex}}c@{\hspace{1ex}}||@{\hspace{1ex}}c@{\hspace{1ex}}|@{\hspace{1ex}}c@{\hspace{1ex}}|@{\hspace{1ex}}c@{\hspace{1ex}}|@{\hspace{1ex}}c@{\hspace{1ex}}}
    \  & g_1 \ & g_2 \ & g_3 \ \\[.5ex] \hline
    Agent~1 \ & 1 \ & 1 \ & 2 \ \\[.5ex] \hline
    Agent~2 \ & 1 \ & 1 \ & 2 \ \\[.5ex] 
\end{array}
\end{eqnarray*}
\end{minipage}
\end{center}

Observe that $g_3$ is twice as valuable than $g_1$ or $g_2$ for both agents. An allocation where one agent gets $\{g_1\}$ and the other gets $\{g_2,g_3\}$ is  EF1 but not EFX. The only possible EFX allocation is where one agent gets $\left\{g_3\right\}$ and the other gets $\left\{g_1,g_2\right\}$, which is clearly fairer than the given EF1 allocation. This example also shows how EFX helps to rule out some unsatisfactory EF1 allocations. Caragiannis et al.~\cite{CaragiannisGravin19} remark that 

\begin{quote}
``\textit{Arguably, EFX is the best fairness analog of envy-freeness of indivisible items.}'' 
\end{quote}

While an EF1 allocation is always guaranteed to exist, very little is known about the existence of EFX allocations. Caragiannis et al.~\cite{CaragiannisKMP016} state that 

\begin{quote}
``\textit{Despite significant effort, we were not able to settle the question of whether an EFX allocation always exists (assuming all items must be allocated), and leave it as an enigmatic open question.}'' 
\end{quote}

Plaut and Roughgarden~\cite{TimPlaut18} show two scenarios for which EFX allocations are guaranteed to exist: $(i)$ All agents have identical valuations (i.e., $v_1=v_2=\dots =v_n$), and $(ii)$ Two agents (i.e., $n=2$). Unfortunately, starting from three agents, even for the well studied class of \emph{additive valuations}, it is open whether EFX allocations exist. Plaut and Roughgarden ~\cite{TimPlaut18} also remark that: 

\begin{quote}
``\textit{The problem seems highly non-trivial even for three players with different additive valuations.}'' 
\end{quote}

Furthermore, it is also suspected in~\cite{TimPlaut18} that EFX allocations may not exist in general settings: 

\begin{quote}
``\textit{We suspect that at least for general valuations, there exist instances where no EFX allocation exists, and it may be easier to find a counterexample in that setting.}'' 
\end{quote}

Contrary to this suspicion, we show that 

\begin{theorem*}
	EFX allocations always exist for three agents with additive valuations.
\end{theorem*}

\paragraph{EFX with \emph{charity}:} Quite recently there have been studies~\cite{CaragiannisGravin19,CKMS20} that consider relaxations of EFX, called ``EFX with charity''. Here we look for partial EFX allocations, where not all items need to be allocated (some of them remain unallocated). There is a trivial such allocation where no item is allocated to any agent.  Therefore, the goal is to determine allocations with some \emph{qualitative} or \emph{quantitative} bound on the set of unallocated items. For instance, Chaudhury et al.~\cite{CKMS20} show how to determine a partial EFX allocation $X$ and a pool of unallocated items $P$ such that no agent envies the pool (i.e. for any agent $i$, we have $v_i(X_i) \geq v_i(P)$), and $P$ has less than $n$ items (i.e., $\lvert P \rvert < n$), even in the case of \textit{general valuations}. In case of additive valuations, Caragiannis et al. ~\cite{CaragiannisGravin19} show the existence of a partial EFX allocation $X = \langle X_1,X_2, \dots , X_n \rangle$, where every agent gets at least half the value of his bundle in the allocation that maximizes the \emph{Nash welfare} i.e., the geometric mean of agents' valuations. (suggesting that unallocated items are not too valuable).

The Nash welfare of a fair allocation is often considered as a measure of its \emph{efficiency} ~\cite{CaragiannisGravin19}: Intuitively, it captures how much \emph{average} welfare the allocation achieves while still remaining fair. The result of Caragiannis et al.~\cite{CaragiannisGravin19} imply that there are \emph{efficient} partial EFX allocations (partial EFX allocations with a 2-approximation of the maximum possible Nash welfare). Indeed, it is a natural question to ask whether there are complete EFX allocations (all items are allocated) with good efficiency. To this end, Caragiannis et al.~\cite{CaragiannisGravin19} conjecture: 

\begin{quote}
``\textit{In particular, we suspect that adding an item to an allocation problem (that provably has an EFX allocation) yields another problem that also has an EFX allocation with at least as high Nash welfare as the initial one.}''\footnote{This was posed as a monotonicity conjecture in their presentation at EC'19.} 
\end{quote}

If this conjecture is true, it implies the existence of an efficient complete EFX allocation. We show (in Section~\ref{counter_example_monotonicity}) that 

\begin{center}
The above conjecture is false. 
\end{center}

To disprove the conjecture  we exhibit an instance where there exists a partial EFX allocation with higher Nash welfare than the Nash welfare of any complete EFX allocation. This also highlights an inherent  barrier in the current techniques to determining EFX allocations: Several of the existing algorithms for approximate EFX allocations (\cite{TimPlaut18}) and EFX allocations with charity (\cite{CKMS20}) start with a inefficient partial  EFX allocation and make it more efficient iteratively by cleverly allocating some of the unallocated items and unallocating some of the allocated items. However, our instance in Section~\ref{counter_example_monotonicity} shows that such approaches will not help if our goal is to determine a  complete EFX allocation. 

A large chunk of our work in this paper develops better tools to overcome this particular barrier, and we consider the tools introduced here as the most innovative technical contribution of our work. We also feel that these tools and the instance may help resolving the major open problem of the existence of EFX allocations for more than three agents and more general valuations (positively or negatively).           

\subsection{Our Contributions}
Our major contribution in this paper is to prove that an EFX allocation always exists when there are three agents with additive valuations. The proof is algorithmic. To discuss our techniques, we first briefly highlight how we overcome two barriers in the current techniques.

\paragraph{Splitting bundles:} We first sketch the simple algorithm of Plaut and Roughgarden~\cite{TimPlaut18} that determines an EFX allocation when all agents have the same valuation function, say $v$. Let us restrict our attention to the special case where there is no zero marginals, i.e., for any $S \subseteq M$ and $g \notin S$ we have  $v(S\cup g) > v(S)$. Also, note that since agents have the same valuation function, if $v(X_i) < v(X_j \setminus g)$ for two agents $i$ and $j$ for some $g \in X_j$ then we have $v(X_{i_{\mathit{min}}}) < v(X_j \setminus g)$ where $i_{\mathit{min}}$ is the agent with the lowest valuation. The algorithm in ~\cite{TimPlaut18} starts off with an arbitrary allocation (not necessarily EFX), and as long as there are agents $i$ and $j$ such that $v(X_i) < v(X_j \setminus g)$ for some $g \in X_j$, the algorithm takes the item $g$ away from $j$ ($j$'s new bundle is $X_j \setminus g$) and adds it to $i_{\mathit{min}}$'s bundle ($i_{\mathit{min}}$'s new bundle is $X_{i_{\mathit{min}}} \cup g$). Also, note that after re-allocation the only changed bundles are that of $i_{\mathit{min}}$ and $j$, and both of them have valuations still higher than $i_{\mathit{min}}$'s initial valuation: $v(X_{i_{\mathit{min}}} \cup g) > v(X_{i_{\mathit{min}}})$ and $v(X_j \setminus g) > v(X_{i_{\mathit{min}}})$. Observe that such an operation increases the valuation of an agent with the lowest valuation. Thus, after finitely many applications of this re-allocation we must arrive at an EFX allocation. Note that this crucially uses the fact that the agents have identical valuations. In the general case, the valuation of agent $j$ may drop significantly after removing $g$ and $j$'s current valuation may be even less than $i_{\mathit{min}}$'s initial valuation.  \emph{Therefore, it is important to understand how agents value item(s) that we move across the bundles.} To this end, we carefully split every bundle into \emph{upper} and \emph{lower} half bundles (see~\eqref{split into upper and lower} in Section~\ref{prelim}). We systematically quantify the agent's relative valuations agents have for these upper and lower half bundles, and in most cases, we are able to move these bundles from one agent to the other, improve the valuation of some of the agents, and  while still guaranteeing EFX property. This idea is detailed in Sections~\ref{threesources} and ~\ref{twosources}.

\paragraph{A new potential function:}  We need to show that there is progress after every swap of half bundles. The typical method here is to show improvement of the valuation vector on the Pareto front (see [CKMS20] and [PR18]). However, there are limitations to this approach: In particular, we show an instance and a partial EFX allocation such that the valuation vector of any complete EFX allocation does not Pareto dominate the valuation vector of the existing partial EFX allocation. To overcome this barrier, we first pick an arbitrary agent $a$ at the beginning and show that whenever we are unable to improve the valuation vector on the Pareto front, we can strictly increase $a$'s valuation. In other words, the valuation of a particular agent $a$ never decreases throughout re-allocations, and it improves after finitely many re-allocations, showing convergence. A more elaborate discussion on this technique is presented in Section~\ref{prelim}. 

\subsection{Further Related Work}
Fair division has received significant attention since the seminal work of Steinhaus~\cite{Steinhaus48} in the 1940s, where he introduced the cake cutting problem among $n>2$ agents. Perhaps the two most crucial notions of fairness properties that can be guaranteed in case of divisible items are \emph{envy-freeness} and \emph{proportionality}. In a proportional allocation, each agent gets at least a $1/n$ share of all the items. In case of indivisible items, as mentioned earlier, none of these two notions can be guaranteed. While EF1 and EFX are fairness notions that relax envy-freeness, the most popular notion of fairness that relaxes proportionality for indivisible items is \emph{maximin share} (MMS), which was introduced by Budish~\cite{budish2011combinatorial}. While MMS allocations do not always exist~\cite{KPW18}, but there has been extensive work to come up with approximate MMS allocations~\cite{budish2011combinatorial,BL16,AMNS17,BK17,KPW18,GhodsiHSSY18,JGargMT19,GargT19}. 

While much research effort goes into finding fair allocations, there has also been a lot of interest in guaranteeing \emph{efficient} fair allocations. A standard notion of efficiency is \emph{Pareto-optimality}\footnote{An allocation $X = \langle X_1,\dots ,X_n \rangle$ is Pareto-optimal if there is no allocation $Y = \langle Y_1, \dots , Y_n \rangle$ where $v_i(Y_i) \geq v_i(X_i)$ for all $i \in [n]$ and $v_j(Y_j) > v_j(X_j)$ for some $j$.}.  Caragiannis et al.\ \cite{CaragiannisKMP016} showed that any allocation that has the maximum Nash welfare is guaranteed to be Pareto-optimal (efficient) and EF1 (fair). Therefore, the Nash welfare of an allocation is also considered as a measure of efficiency and fairness of an allocation. However, finding an allocation with the maximum Nash welfare is APX-hard~\cite{Lee17}, and its approximation has received a lot of attention recently, e.g.,~\cite{ColeG18,ColeDGJMVY17,AnariGSS17,GargHM18,AnariMGV18,BKV18,ChaudhuryCGGHM18,GargKK20}.  Barman et al.~\cite{BKV18} give a pseudopolynomial algorithm to find an allocation that is both EF1 and Pareto-optimal. Other works try to approximate MMS with Pareto-optimality~\cite{GargM19} or explore relaxations of EFX with high Nash welfare~\cite{CaragiannisGravin19}. 

\paragraph{Applications:} There are several real-world scenarios where resources need to be divided fairly and efficiently, e.g., splitting rent among tenants, dividing inheritance property in a family, splitting taxi fares among riders, and many more. One examples of fair division techniques used in practice is Spliddit (\url{http://www.spliddit.org}).  Since its launch in 2014, Spliddit has had several thousands of users~\cite{CaragiannisKMP016}. For more details on Spliddit, we refer the reader to \cite{GP14,TimPlaut18}. Another example is \emph{Course Allocate}, which  is  used by the Wharton School at the University of Pennsylvania to fairly allocate 350 courses to 1700 MBA students~\cite{TimPlaut18, BudishCKO17}. Kurokawa et al.\ \cite{KurokawaPS18} used \emph{leximin fairness} to allocate unused classrooms in public schools to charter schools in California. The best part of the allocations determined in all these applications is that they yield results that not only \emph{seem} fair on most instances, but also come with mathematical guarantees.

\section{Preliminaries and Technical Overview}\label{prelim}
An instance $I$ of fair allocation problem is a triple $\langle [3], M, \mathcal{V} \rangle$, where we have three agents $1$, $2$, and $3$, a set $M$ of $m$ indivisible items {(or goods)}, and a set of valuation functions $\mathcal{V} = \left\{v_1,v_2,v_3\right\}$, where each $v_i \colon 2^M \rightarrow \mathbb{R}_{\geq 0}$ captures the utility agent $i$ has for all the different subsets of goods that can be allocated. We assume that the valuation functions are \emph{additive} ($v_i(S) = \sum_{g \in S} v_i(\left\{g\right\})$) and \emph{normalized} ($v_i(\emptyset)= 0$). For ease of notation, we write $v_i(g)$ for $v_i(\left\{g\right\})$. Further, we write $S \oplus_i T$ for $v_i(S) \oplus v_i(T)$ with $\oplus \in \left\{\leq, \geq ,< , > \right\}$. Given an allocation $X = \langle X_1,X_2, \dots ,X_n \rangle$ we say that $i$ \emph{strongly envies} a bundle $S \subseteq M$ if $X_i <_i S \setminus g$ for some $g \in S$, and we say that $i$ \emph{weakly envies} $S$ if $X_i <_i S$ but $X_i \geq_i S \setminus g$ for all $g \in S$. From this perspective an allocation is an EFX allocation if and only if no agent strongly envies another agent.

\paragraph{Non-degenerate instances:} We call an instance $I = \langle [3], M, \mathcal{V} \rangle$ non-degenerate if and only if no agent values two different sets equally, i.e., $\forall i \in [3]$ we have $v_i(S) \neq v_i(T)$ for all $S \neq T$. We first show that it suffices to deal with non-degenerate instances. Let $M = \left\{g_1,g_2,\dots,g_m\right\}$. We perturb any instance $I$ to $I(\varepsilon) = \langle [3],M ,\mathcal{V}(\varepsilon) \rangle$, where for every $v_i \in \mathcal{V}$ we define $v'_i \in \mathcal{V}(\varepsilon)$, as

$$ v'_i(g_j) = v_i(g_j) + \varepsilon 2^{j}.$$ 

\begin{lemma}
  \label{non-degeneracy-technical} \label{non-degeneracy-main}
  Let $\delta = \min_{i \in [3]} \min_{S,T \colon v_i(S) \neq v_i(T)} \abs{ v_i(S) - v_i(T)}$ and let $\varepsilon > 0$ be such that $\varepsilon \cdot 2^{m+1}  < \delta$. Then
  \begin{enumerate}
  \item For any agent $i$ and $S,T \subseteq M$ such that $v_i(S) > v_i(T)$, we have $v'_i(S) > v'_i(T)$.
  \item $I(\varepsilon)$ is a non-degenerate instance. Furthermore, if $X = \langle X_1,X_2,X_3 \rangle$ is an EFX allocation for $I(\varepsilon)$ then $X$ is also an EFX allocation for $I$.
    \end{enumerate}
\end{lemma}
\begin{proof}
For the first statement of the lemma, observe that 
\begin{align*}
 v'_i(S) - v'_i(T)  &=  v_i(S) - v_i(T)  + \varepsilon(\sum_{g_j \in S \setminus T}2^j - \sum_{g_j \in T \setminus S}2^j) \\
                                   &\geq \delta -  \varepsilon \sum_{g_j \in T \setminus S}2^j\\
                                   &\geq \delta -  \varepsilon \cdot (2^{m+1}-1)\\
                                   &>0 
\end{align*}

For the second statement of the lemma, consider any two sets $S,T \subseteq M$ such that $S \neq T$. Now, for any $i \in [3]$, if $v_i(S) \neq v_i(T)$, we have $v'_i(S) \neq v'_i(T)$ by the first statement of the lemma. If $v_i(S) = v_i(T)$, we have $v'_i(S) - v'_i(T) = \varepsilon(\sum_{g_j \in S \setminus T}2^j - \sum_{g_j \in T \setminus S}2^j) \neq 0$ (as $S \neq T$). Therefore, $I(\varepsilon)$ is non-degenerate.
 
For the final claim, let us assume that $X$ is an EFX allocation in $I(\varepsilon)$ and not an EFX allocation in $I$. Then there exist $i,j$, and $g \in X_j$ such that $v_i(X_j \setminus g) > v_i(X_i)$. In that case, we have $v'_i(X_j \setminus g) > v'_i(X_i)$ by the first statement of the lemma, implying that $X$ is not an EFX allocation in $I(\varepsilon)$ as well, which is a contradiction. 
\end{proof}

\textit{From now on we only deal with non-degenerate instances}. {In non-degenerate instances, all goods have positive value for all agents.}

\paragraph{Overall approach:} An allocation $X'$ \emph{Pareto dominates} an allocation $X$ if $v_i(X_i) \le v_i(X_i')$ for all $i$ with strict inequality for at least one $i$. The existing algorithms for ``EFX with charity'' \cite{CKMS20} or ``approximate EFX allocations'' \cite{TimPlaut18} construct a sequence of EFX allocations in which each allocation Pareto dominates its predecessor. However we exhibit in Section~\ref{counter_example_monotonicity} a partial EFX allocation that is not Pareto dominated by any complete EFX allocation. Thus we need a more flexible approach in the search for a complete EFX allocation. 

We name the agents $a$, $b$, and $c$ arbitrarily and consider the lexicographic ordering of the triples
\[    \phi(X) = (v_a(X_a), v_b(X_b), v_c(X_c)), \]
i.e., $\phi(X) \prec_{\mathit{lex}} \phi(X')$ ($X'$ \emph{dominates} $X$) if (i) $v_a(X_a) < v_a(X'_a)$ or (ii) $v_a(X_a) = v_a(X'_a)$ and $v_b(X_b) < v_b(X'_b)$ or (iii) $v_a(X_a) = v_a(X'_a)$ and $v_b(X_b) = v_b(X'_b)$ and $v_c(X_c) < v_c(X'_c)$.
We construct a sequence of allocations in which each allocation dominates its predecessor. Of course, if $X'$ Pareto dominates $X$, then it also dominates $X$, so we can use all the update rules in~\cite{CKMS20}. 

Our goal then is to iteratively construct a sequence of EFX allocations such that each EFX allocation dominates its predecessor. %We will construct the next allocation $X'$ such that $X'_i \subseteq X_1 \cup X_2 \cup X_3 \cup g$ for $i \in [3]$ and $X'$ dominates $X$.}

\paragraph{Most envious agent:} We use the notion of a most envious agent, introduced in~\cite{CKMS20}. Consider an allocation $X$, and a set $S \subseteq M$ that is envied by at least one agent. For an agent $i$ such that $S >_i X_i$, we ``measure the envy'' that agent $i$ has for $S$ by $\kappa_X(i,S)$, where $\kappa_X(i,S)$ is the size of a smallest subset of $S$ that $i$ still envies, i.e., $\kappa_X(i,S)$ is the smallest cardinality of a subset $S'$ of $S$ such that $S' >_i X_i$. Thus, the smaller the value of $\kappa_X(i,S)$, the greater the envy of agent $i$ for the set $S$. So let $\kappa_X(S) = \mathit{min}_{i \in [3]} \kappa_X(i,S)$. Naturally, we define the set of the \emph{most envious agents} $A_X(S)$ for a set $S$ as the set of agents with smallest values of $\kappa_X(i,S)$, i.e.,
\[      A_X(S) = \left\{i \mid S >_i X_i \text{ and } \kappa_X(i,S) = \kappa_X(S) \right\}.\]
The following simple observation about the most envious agents of specific kinds of bundles will be useful. 

\begin{observation}
\label{A_X_neverempty}
Given any allocation $X$, and an unallocated good $g$, for any $i \in [3]$, $A_X(X_i \cup g) \neq \emptyset$.
\end{observation}

\begin{proof}
 It suffices to prove that there exists at least one agent who strictly prefers $X_i \cup g$  over his own bundle in allocation $X$. This is guaranteed since we are dealing with non-degenerate instances, in which $X_i \cup g >_i X_i$. 
\end{proof}

\paragraph{Champions and Champion Graph $M_X$:} Let $X$ be the partial EFX allocation at any stage in our algorithm, and let $g$ be an unallocated good. We say that $i$ \emph{champions} $j$ (w.r.t $g$) if $i$ is a most envious agent for $X_j \cup g$, i.e., $ i \in A_X(X_j \cup g)$. We define the \emph{champion graph} $M_X$, where each vertex corresponds to an agent and there is a directed edge $(i,j) \in M_X$ if and only if $i$ champions $j$. 

\begin{observation}
The champion graph $M_X$ is cyclic.
\end{observation}

\begin{proof}
By Observation~\ref{A_X_neverempty}, we have that the set of champions of any agent is never empty. Therefore, every vertex in $M_X$ has at least one incoming edge. Thus $M_X$ is cyclic.
\end{proof}

\emph{If $i$ champions $j$, we define $G_{ij}$ as a largest cardinality subset of $X_j \cup g$ such that $(X_j \cup g) \setminus G_{ij} >_i X_i$.} Since the valuations are additive, note that such a subset can be identified efficiently as the set $K$ of the $k$ least valuable goods for $i$ in $X_j \cup g$ such that $(X_j \cup g) \setminus K >_i X_i$ and $k$ is maximum. Now we make some small observations.

\begin{observation}
\label{nobodyenvieschampion}
 Assume $i$ champions $j$. 
 \begin{enumerate}
   \item We have $((X_j \cup g) \setminus G_{ij}) \setminus h \leq_k X_k$ for all $h \in (X_j \cup g) \setminus G_{ij}$ and all agents $k$ including $i$. 
   \item If agent $k$ does not champion $j$, we have $(X_j \cup g) \setminus G_{ij} \leq_k X_k$.
 \end{enumerate}  
\end{observation}

\begin{proof}
Note that by definition, $G_{ij}$ is a largest cardinality subset of $X_j \cup g$ such that $i$ values $(X_j \cup g) \setminus G_{ij}$ more than $X_i$. This implies that $(X_j \cup g) \setminus G_{ij}$ is a smallest cardinality subset of $X_j \cup g$ that $i$ values more than $X_i$. Thus $\lvert (X_j \cup g) \setminus G_{ij} \rvert = \kappa_X(i,X_j \cup g)$. Since $i$ champions $j$, we have that $ i \in A_X(X_j \cup g)$ and thus $\kappa_X(i,X_j \cup g) = \kappa_X(X_j \cup g)$. Now, no agent $k$ values a subset of $X_j \cup g$ of size less than $\kappa_X(k, X_j \cup g)$ more than $X_k$. Note that $((X_j \cup g) \setminus G_{ij}) \setminus h$ has size $\kappa_X(X_j \cup g) - 1 < \kappa_X(k,X_j \cup g)$ and ,thus, $((X_j \cup g) \setminus G_{ij}) \setminus h \leq_k X_k$.

Now if $k$ did not champion $j$ then $\kappa_X(k,X_j \cup g) < \kappa_X(X_j \cup g)$. Thus, $\lvert (X_j \cup g) \setminus G_{ij} \rvert = \kappa_X(X_j \cup g) < \kappa_X(k,X_j \cup g)$. Since $k$ values any subset of $X_j \cup g$ of size less than $\kappa_X(k,X_j \cup g)$ at most $X_k$, we have $(X_j \cup g) \setminus G_{ij} \leq_k X_k$.
\end{proof}

We next mention two cases where it is known how to obtain a {Pareto} dominating EFX allocation from an existing EFX allocation. For an allocation $X$, we define the \emph{envy graph} $E_X$, whose vertices represent agents, and in which there is a directed edge from $i$ to $j$ if $i$ envies $j$, i.e., $X_j >_i X_i$. We can assume without loss of generality (w.l.o.g.) that $E_X$ is acyclic.

\begin{fact}[\cite{LiptonMMS04}]
\label{acyclic-envy-graph}
Let $X = \langle X_1,X_2,X_3 \rangle $ be an EFX allocation. Then there exists another EFX allocation $Y = \langle Y_1,Y_2,Y_3 \rangle$, where for all $i\in [3]$, $Y_i = X_j$ for some $j \in [3]$, such that $E_Y$ is acyclic and
{$\phi(Y) \succeq_{\mathit{lex}} \phi(X)$} (because $Y$ Pareto dominates $X$).
\end{fact}

\begin{observation}[\cite{CKMS20}]
\label{champion_in_subtree} 
Consider an EFX allocation $X$. Let $s$ be any agent and let $g$ be an unallocated good. If $i$ champions $s$ and $i$ is reachable from $s$ in $E_X$, then there is an EFX allocation $Y$ Pareto dominating $X$. Additionally, agent $s$ is strictly better off in $Y$, i.e., $Y_s >_s X_s$.
\end{observation}

\begin{proof}
 We have that $i$ is reachable from $s$ in $E_X$. Let $t_1 \rightarrow t_2 \rightarrow \dots \rightarrow t_k$ be the path from $t_1 = s$ to $t_k = i$ in $E_X$. We determine a new allocation $Y$ as follows:
\begin{align*}
Y_{t_j} &= X_{t_{j+1}} &\text{ for } j \in [k-1]\\
Y_i &= (X_s \setminus G_{is}) \cup g\\
Y_{\ell} &= X_{\ell} &\text{for all other }\ell 
\end{align*}
Note that every agent along the path has strictly improved his valuation: Agents $t_1$ to $t_{k-1}$ got bundles they envied in $E_X$ and agent $i$ championed $s$ and got $(X_s \setminus G_{is} \cup g)$, which is more valuable to $i$ than $X_i$ (by definition of $G_{is}$). Also, every other agent retained their previous bundles and thus their valuations are not lower than before. Thus $\phi(Y) \succ_{\mathit{lex}} \phi(X)$ and also $Y_s >_s X_s$ ($s$ was an agent along the path). It only remains to argue that $Y$ is EFX. To this end, consider any two agents $j$ and $j'$. We wish to show that $j$ does not strongly envy $j'$ in $Y$. 
\begin{description}
\item[Case ${j' \neq i}$:] Note that $Y_{j'} = X_{\ell}$ for some $\ell \in [3]$ ($j'$ either received a bundle of another agent when we shifted the bundles along the path or  retained the previous bundle). Also, note that $Y_j \geq_j X_j$ (no agent is worse off in $Y$). Therefore,  $Y_j \geq_j X_j \geq_j X_{\ell} \setminus h =_j Y_{j'} \setminus h$ for all $h \in Y_{j'}$ ($j$ did not strongly envy $\ell$ in $X$ as $X$ was EFX).

\item[Case ${j' = i}$:] We have $Y_{j'} = (X_s \setminus G_{is}) \cup g$. Since $i$ championed $s$, by Observation~\ref{nobodyenvieschampion} (part 1) we have that $((X_s \setminus G_{is}) \cup g) \setminus h \leq_j X_j $. Like earlier, $Y_j \geq_j X_j$ (no agent is worse off in $Y$). Thus $j$ does not strongly envy $i$. \qedhere
\end{description}
\end{proof}

Observation~\ref{champion_in_subtree} implies that if there is some unallocated good and (i) if the envy graph $E_X$ has a single source\footnote{A source is a vertex in $E_X$ with in-degree zero.} or (ii) any agent champions himself then there is a {strictly} {Pareto} dominating EFX allocation.

\begin{corollary}
\label{singlesource}
 Let $X$ be an EFX allocation, and $g$ be an unallocated good. If $E_X$ has a single source $s$, or $M_X$ has a $1$-cycle involving agent $s$, then there is an EFX allocation $Y$ that Pareto dominates $X$ in which $Y_s >_s X_s$.
\end{corollary}

\begin{proof}
 If $E_X$ has a single source $s$, the champion of $s$ (which always exist, by Observation~\ref{A_X_neverempty}) is reachable from $s$. If $M_X$ has a $1$-cycle involving agent $s$ then again the champion of $s$ (which is $s$ itself) is reachable from $s$. In both cases, since the champion of $s$ is reachable from $s$ in the envy graph $E_X$, there is a Pareto dominating allocation $Y$ such that $Y_s >_s X_s$ by Observation~\ref{champion_in_subtree}.   
\end{proof}

\textit{Hence, starting from Section~\ref{threesources}, we only discuss the cases where the envy-graph has more than one source and there are no self-champions.}

We start with some simple yet crucial observations.

\begin{observation}
\label{lowerhalflessvaluable}
 If $i$ champions $j$ and $X_i \geq_i X_j$, then $g \notin G_{ij}$, $G_{ij} \subseteq X_j$, and $G_{ij} <_i g$.
\end{observation}

\begin{proof}
 We have  $i \in A_X(X_j \cup g)$. Since $g \notin X_j$,  $G_{ij} \subseteq X_j \cup g$, and valuations are additive and we have that $v_i((X_j \cup g) \setminus G_{ij}) = v_i(X_j) + v_i(g) - v_i(G_{ij})$. Again since $i \in A_X(X_j \cup g)$, by the definition of $G_{ij}$, $(X_j \cup g)\setminus G_{ij} >_i X_i$, and hence, $v_i(X_i) < v_i(X_j) + v_i(g) - v_i(G_{ij})$. Now we have $X_i \geq_i X_j$, implying that $G_{ij} <_i g$, and therefore, $g \not\in G_{ij}$.
\end{proof}

Observation~\ref{lowerhalflessvaluable} tells us that if $i$ champions $j$, and $i$ does not envy $j$, then $G_{ij} \subseteq X_j$. Therefore, we can split the bundle of agent $j$ into two parts $G_{ij}$ and $X_j \setminus G_{ij}$. We refer to $G_{ij}$ as the \emph{lower-half bundle} of $j$, and to $X_j \setminus G_{ij}$ as the \emph{upper-half bundle} of $j$, and visualize the bundle of agent $j$ as
\begin{equation}\label{split into upper and lower} X_j = \myBox{X_j \setminus G_{ij}}{G_{ij}}{j}   \qquad \parbox{0.5\textwidth}{\text{if $i$ champions $j$ and $i$ does not envy $j$.}} \end{equation}

We collect some more facts about the values of lower and upper half bundles.

% \paragraph{Values of lower halves:} We relate the value of $g$ to the values of lower halves. First, if $i$ does not envy bundle $X_j$ but is a champion for $j$ ($i \in A_X(X_j \cup g)$), then $i$ values $G_{ij}$ less than $g$. 

% Second, if $i$ is a champion of agent $j$ and $j$ is not a self-champion then $j$ values $G_{ij}$ at least as high as $g$. 

\begin{observation}
\label{lowerhalfmorevaluable}
 If $i$ champions $j$ and $j$ does not champion himself (self-champion), then we have $G_{ij} \not= \emptyset$ and $G_{ij} \geq_j g$.
\end{observation}

\begin{proof}
Since $j$ does not self-champion, by Observation~\ref{nobodyenvieschampion} (part 2), we have that $(X_j \cup g)\setminus G_{ij} \leq_j X_j$. Since $g \notin X_j$ and $G_{ij} \subseteq X_j \cup g$ we have $v_j((X_j \cup g) \setminus G_{ij}) = v_j(X_j) +v_j( g) - v_j(G_{ij}) \leq v_j(X_j)$, implying that $G_{ij} \geq_j g$. Since the value of $g$ for $j$ is non-zero, $G_{ij}$ is non-empty. 
\end{proof}

% \paragraph{Values of upper halves$\colon$} Similarly now we quantify the valuation the most envious agents have for the corresponding upper half bundles.

\begin{observation}
\label{upperhalfvaluable}
 Let $i$ champion $j$, and $X_i \geq_i X_j$. Let $i'$ champion $k$ and $X_{i'} \geq_{i'} X_k$. If $i$ does not champion $k$, then $X_j \setminus G_{ij} >_i X_{k} \setminus G_{i'k}$.
\end{observation}

\begin{proof}
Since $i \in A_X(X_j \cup g)$ and $X_i \geq_i X_j$, by Observation~\ref{lowerhalflessvaluable}, we have $g \notin G_{ij}$. Thus, $G_{ij} \subseteq X_j$. By the same reasoning, $g \notin G_{i'k}$ and $G_{i'k} \subseteq X_k$. Therefore, $(X_j \cup g) \setminus G_{ij} = (X_j \setminus G_{ij}) \cup g$, and $(X_k \cup g) \setminus G_{i'k} = (X_k \setminus G_{i'k}) \cup g$. By the definition of $G_{ij}$, we have $ (X_j \setminus G_{ij}) \cup  g >_i X_i$. Since $i \notin A_X(X_k \cup g)$, we have $X_i \geq_i (X_k \setminus G_{i'k}) \cup g$ by Observation~\ref{nobodyenvieschampion} (part 2). Combining the two inequalities, we have $(X_j \setminus G_{ij}) \cup g >_i (X_k \setminus G_{i'k}) \cup g$, which implies $X_j \setminus G_{ij} >_i X_{k} \setminus G_{i'k}$. 
\end{proof}

In the upcoming sections, we show how to derive a dominating EFX allocation from an existing EFX allocation. Corollary~\ref{singlesource} already deals with the cases that $E_X$ has a single source or $M_X$ has a 1-cycle.
\emph{We proceed under the following general assumptions: $E_X$ is cycle-free and has at least two sources and there is no 1-cycle in $M_X$.} We  distinguish the remaining cases by the number of sources in $E_X$.

\section{Existence of EFX: Three sources in $E_X$}
\label{threesources}

If $E_X$ has three sources, the allocation $X$ is envy-free, i.e., $X_i \geq_i X_j$ for all $i$ and $j$.  We make a case distinction by whether or not $M_X$ contains a $2$-cycle.  

\subsection{$2$-cycle in $M_X$}
\label{twocycle_in_M_X}
Assume without loss of generality that agent 2 champions agent 1 and agent 1 champions agent 2. Since $X_1 \geq_1 X_2$ and $X_2 \geq_2 X_1$, the bundles $X_1$ and $X_2$ decompose according to \lref{split into upper and lower}. Since  neither 1 nor 2 self-champion (as $M_X$ has no $1$-cycle), by Observation~\ref{upperhalfvaluable}, we have $X_2 \setminus G_{12} >_1 X_1 \setminus G_{21}$ and $X_1 \setminus G_{21} >_2 X_1 \setminus G_{12}$. We swap the upper-halves of $X_1$ and $X_2$ to obtain
\[ X' = \Allocationfive{X_2 \setminus G_{12}}{{G_{21}}}{X_1 \setminus G_{21}}{G_{12}}{X_3}.\]
Note that agent 3 has the same valuation as before, while 1 and 2 are strictly better off. If $X'$ is EFX we are done. So assume otherwise. We first determine the potential strong envy edges.

\begin{itemize}
    \item \textit{From 1}: We replaced the more valuable (according to 1) $X_2 \setminus G_{12}$ in $X_2$ with the less valuable $X_1 \setminus G_{21}$ and left $X_3$ unchanged. Thus 1 is strictly better off and according to him, the valuations of the bundles of 2 and 3 in $X'$ is at most the valuation of their bundles in $X$. As 1 did not envy 2 and 3 before in $X$, 1 does not envy 2 and 3 in $X'$. 
    \item \textit{From 2}: A symmetrical argument shows that 2 does not envy 1 and 3. 
    \item \textit{From 3}: For agent 3, the sum of the valuations of agents 1 and 2 has not changed by the swap and 3 envied neither 1 nor 2 before the swap. Thus 3 envies at most one of the agents 1 and 2 after the swap. Assume without loss of generality that he envies agent 2. We then replace the lower-half bundle of agent 2 ($G_{12}$) with $g$ to obtain
\[ X'' = \Allocationfive{X_2 \setminus G_{12}}{G_{21}}{X_1 \setminus G_{21}}{g}{X_3}.\]

In $X''$, agent 2 is still strictly better off than in $X$ since by the definition of $G_{21}$, we have $(X_1 \setminus G_{21}) \cup g >_2 X_2$. Thus, $X''$ Pareto dominates $X$. We still need to show that $X''$ is  EFX. To this end, observe that as we have not changed the bundles of agents 1 and 3, there is no strong envy between them. So we only need to exclude strong envy edges to and from agent 2. 
   \begin{itemize}
       \item \textit{Nobody strongly envies agent 2}: Note that 2 championed 1. Thus, $((X_1 \setminus G_{21}) \cup g) \setminus h \leq_1 X_1$ and $((X_1 \setminus G_{21}) \cup g) \setminus h \leq_3 X_3$ for all $h \in (X_1 \setminus G_{21}) \cup g$ by Observation~\ref{nobodyenvieschampion} (part 1). Since both 1 and 3 are not worse off than before, they do not strongly envy 2. 
       \item \textit{Agent 2 does not envy anyone}: We have that $(X_1 \setminus G_{21}) \cup g >_2 X_2$. Also according to 2, the valuation of the current bundles of 1 and 3 is at most their previous one, and 2 did not envy them before (when he had $X_2$). Hence, 2 does not envy 1 and 3.
   \end{itemize}
 \end{itemize}

 We have thus shown that $X''$ is EFX and Pareto dominates $X$. Actually, the strategy described above handles a more general situation. It yields a Pareto dominating EFX allocation as long as 3 envies neither 1 nor 2 initially,  even if 1 and 2 envied (\emph{not strongly envied}) 3  initially: %We formalize the statement in the following remark.

\begin{remark}
\label{generalizing-2cycle-technique}
Let $X$ be an EFX allocation, and let $g$ be an unallocated good. If $M_X$ has a $2$-cycle, say involving agents 1 and 2, and agent 3 envies neither 1 nor 2, then there exists an EFX allocation $Y$ Pareto dominating $X$.
\end{remark} 

Remark~\ref{generalizing-2cycle-technique} will be helpful when we deal with certain instances where $E_X$ has two sources later in Section~\ref{twosources}.

\subsection{No $2$-cycle in $M_X$}
\label{notwocycle_in_M_X}

We now consider the case when $M_X$ has no two cycle. Since $M_X$ is cyclic and we neither have a $1$-cycle nor a $2$-cycle, we must have a $3$-cycle. Let us assume w.l.o.g. that agent $i+1$ is the unique champion of agent $i$ (indices are modulo 3, so $i+1$ corresponds to $(i \bmod 3)+1$). Since, in addition, $i+1$ does not envy $i$, all three bundles decompose according to \lref{split into upper and lower} and the current allocation can be written as
\[ X = \Allocation{X_1 \setminus G_{21}}{G_{21}}{X_2 \setminus G_{32}}{G_{32}}{X_3 \setminus G_{13}}{G_{13}}.\]

Let us collect what we know for agent 1's valuation of the upper-half bundles: 1 uniquely champions 3, while 2 and 3 uniquely champion 1 and 2, respectively. Also, the current allocation is envy-free. Thus $X_i \geq X_j$ for all $i,j \in [3]$. By Observation~\ref{upperhalfvaluable}, we know that $X_3 \setminus G_{13} >_1 \max_1(X_1 \setminus G_{21}, X_2 \setminus G_{32})$\footnote{$\max_1(X_1 \setminus G_{21}, X_2 \setminus G_{32})$ is 1's favorite bundle out of $X_1 \setminus G_{21}$ and $X_2 \setminus G_{32}$} ($X_3 \setminus G_{13}$ is 1's favorite upper-half bundle).

Now, let us collect what we know for agent 1's valuation of the lower-half bundles: 1 champions 3 and does not envy 3's bundle. Thus, by Observation~\ref{lowerhalflessvaluable}, $G_{13} <_1 g$ and $g \not\in G_{13}$. Also, 1 does not champion himself, and 3 champions 1. Thus, by Observation~\ref{lowerhalfmorevaluable}, $g \leq_1 G_{21}$.
We can make similar statements about agents 2 and 3. Since $g \not\in G_{21}$, and our instance is assumed to be non-degenerate, we even have $g <_1 G_{21}$. Tables~\ref{orderingupperhalf} and~\ref{orderinglowerhalf} summarize this information. 

\begin{table}
\begin{center}
\begin{tabular}{ | m{5em} | m{6.7cm}|} 
\hline
Agent 1 & $X_3 \setminus G_{13} >_1 \max_1(X_1 \setminus G_{21},X_2 \setminus G_{32})$ \\ 
\hline
Agent 2 & $X_1 \setminus G_{21} >_2 \max_2(X_2 \setminus G_{32},X_3 \setminus G_{13})$ \\ 
\hline
Agent 3 & $X_2 \setminus G_{32} >_3 \max_3(X_3 \setminus G_{13},X_1 \setminus G_{21})$ \\ 
\hline
\end{tabular}
\end{center}
\caption{No 2-cycle in $M_X$: Ordering for the upper half bundles. }
\label{orderingupperhalf}
\end{table}

\begin{table}
\begin{center}
\begin{tabular}{ | m{5em} | m{3cm}|} 
\hline
Agent 1 & $G_{21} >_1 g >_1 G_{13}$\\ 
\hline
Agent 2 & $G_{32} >_2 g >_2 G_{21}$ \\ 
\hline
Agent 3 & $G_{13} >_3 g >_3 G_{32}$ \\ 
\hline
\end{tabular}
\end{center}
\caption{No 2-cycle in $M_X$: Ordering for the lower half bundles. Furthermore, $g \not\in G_{13}$, $g\not\in G_{21}$, and $g\not\in G_{32}$. }
\label{orderinglowerhalf}
\end{table}

We first move to an allocation where everyone gets their favorite upper-half bundle (we achieve this by  performing a cyclic shift of the upper-half bundles). Thus, the new allocation is:
\[ X' = \Allocation{X_3 \setminus G_{13}}{G_{21}}{X_1 \setminus G_{21}}{G_{32}}{X_2 \setminus G_{32}}{G_{13}}\]
Clearly, every agent is strictly better off, and thus, $X'$ Pareto dominates $X$. If $X'$ is  EFX, we are done. So we assume otherwise. What envy edges could exist? We first observe that no agent will envy the agent from whom it took its upper-half during the cyclic shift.

\begin{observation}
In $X'$, agent $i+1$ does not envy agent $i$ for all $i \in [3]$ (indices are modulo 3).
\end{observation}

\begin{proof}
We just show the proof for $i=1$, and the other cases follow symmetrically. Note that 2 values its current upper-half more than 1's upper-half (it has its favorite upper-half): $X_1 \setminus G_{21} >_2 X_3 \setminus G_{13}$. Similarly 2's also values its lower-half more than 1's lower-half: $G_{32} \geq_2 g >_2 G_{21}$. Therefore, 2 values his entire bundle more than 1's bundle, and hence does not envy 1. 
\end{proof}

Therefore, the only envy edges (and hence \emph{strong} envy edges) can be from agent $i$ to agent $i+1$ as shown in the following figure.\footnote{In the figures that follow, we use red edges to indicate strong envy, and blue edges to indicate weak envy.} 

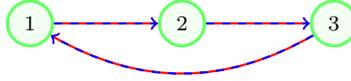
\begin{center}
\begin{tikzpicture}
[
agent/.style={circle, draw=green!60, fill=green!5, very thick},
good/.style={circle, draw=red!60, fill=red!5, very thick, minimum size=1pt},
]

%Vertices for config1
\node[agent]      (a1) at (0,0)      {$\scriptstyle{1}$};
\node[agent]      (a2) at (2,0)      {$\scriptstyle{2}$};
\node[agent]      (a3) at (4,0)     {$\scriptstyle{3}$};

%Vertices for config2
%\node[agent]      (b1) at (5,0)      {$\scriptstyle{1}$};
%\node[agent]      (b2) at (7,0)      {$\scriptstyle{2}$};
%\node[agent]      (b3) at (5,-2.5)     {$\scriptstyle{3}$};

%Edges of G_X
%\draw[->,blue,very thick] (a1)--(a3);
%\draw[->,blue,very thick] (b1)--(b3);

%Edges of M_X -Config1
%\draw[->,red,thick] (a1)--(a2);

%Edges of M_X Config2
\draw[->,red,thick] (a1) -- (a2);
\draw[->,blue,dashed, thick] (a1) -- (a2);

\draw[->,red,thick] (a2) -- (a3);
\draw[->,blue,dashed,thick] (a2) --(a3);

\draw[->,red, thick] (a3) edge[bend left=30] (a1);
\draw[->,blue,dashed, thick] (a3) edge[bend left=30] (a1);

\end{tikzpicture}

\end{center}

\noindent
We now distinguish two cases depending on the number of such strong envy edges.

\paragraph{Three strong envy edges:} In this case, the envy-graph is a 3-cycle. We perform a cyclic shift of the bundles and obtain an EFX allocation Pareto dominating the initial allocation $X$. 

\paragraph{At most two strong envy edges:} Note that in this case, there is a strong envy edge from at least one agent $i \in [3]$ to $i+1$ and there is no strong envy edge from at least one agent $j \in [3]$ to $j+1$. Let us assume without loss of generality that there is a strong envy edge from 1 to 2 , there may or may not be a strong envy edge from 2 to 3, and there is no strong envy edge from 3 to 1.

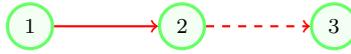
\begin{center}
\begin{tikzpicture}
[
agent/.style={circle, draw=green!60, fill=green!5, very thick},
good/.style={circle, draw=red!60, fill=red!5, very thick, minimum size=1pt},
]

%Vertices for config1
\node[agent]      (a1) at (0,0)      {$\scriptstyle{1}$};
\node[agent]      (a2) at (2,0)      {$\scriptstyle{2}$};
\node[agent]      (a3) at (4,0)     {$\scriptstyle{3}$};

%Vertices for config2
%\node[agent]      (b1) at (5,0)      {$\scriptstyle{1}$};
%\node[agent]      (b2) at (7,0)      {$\scriptstyle{2}$};
%\node[agent]      (b3) at (5,-2.5)     {$\scriptstyle{3}$};

%Edges of G_X
%\draw[->,blue,very thick] (a1)--(a3);
%\draw[->,blue,very thick] (b1)--(b3);

%Edges of M_X -Config1
%\draw[->,red,thick] (a1)--(a2);

%Edges of M_X Config2
\draw[->,red,thick] (a1) -- (a2);
\draw[->,red,dashed,thick] (a2) --(a3);
%\draw[->,blue,dashed thick] (a3) -- (a1);

\end{tikzpicture}
\end{center}

\noindent
Note that 1 is strictly better off in $X'$ than in $X$. The existence of envy from 1 and 2, despite this improvement, allows us to say more about the preference ordering of the upper-half and the lower-half bundles.

\begin{observation}
\label{orderingfor1}
If 1 envies 2 in $X'$, $X_1 \setminus G_{21} >_1 X_2 \setminus G_{32}$, and $G_{32} >_1 G_{21}$.
\end{observation}

\begin{proof}
 We argue by contradiction. Therefore, assume that i.e. $X_1 \setminus G_{21} \leq_1 X_2 \setminus G_{32}$ or $G_{32} \leq_1 G_{21}$. If $X_1 \setminus G_{21} \leq_1 X_2 \setminus G_{32}$, then
\begin{align*}
 (X_1 \setminus G_{21}) \cup G_{32} &\leq_1 (X_2 \setminus G_{32}) \cup G_{32} \\ 
                                  &= X_2 \\
                                  &\leq_1 X_1  &(\text{since 1 did not envy 2 before})\\
                                  &<_1 (X_3 \setminus G_{13}) \cup G_{21} &(\text{since 1 is better off than before})
\end{align*} 
 implying that 1 does not envy 2, a contradiction. If $G_{32} \leq_1 G_{21}$, then 
\begin{align*}
 (X_1 \setminus G_{21}) \cup G_{32} &\leq_1 (X_1 \setminus G_{21}) \cup G_{21} \\
                                  &= X_1 \\
                                  &<_1 (X_3 \setminus G_{13}) \cup G_{21} &(\text{since 1 is better off than before})
\end{align*}
again implying that 1 does not envy 2, a contradiction.
\end{proof}

So we now have
\begin{equation}   X_2 \setminus G_{32} <_1 X_1 \setminus G_{21} <_1 X_3 \setminus G_{13}
  \quad\text{and}\quad
  G_{13} <_1 g <_1 G_{21} <_1 G_{32}.
\end{equation}

We replace the lower-half bundle of 2 ($G_{32}$) by $g$ to obtain
\[
X'' = \Allocation{X_3 \setminus G_{13}}{G_{21}}{X_1 \setminus G_{21}}{g}{X_2 \setminus G_{32}}{G_{13}}.\]

Note that agents 1 and 3 are still strictly better off (as we have not changed their bundles after the cyclic shift of the upper-half bundles) than in $X$. Agent 2 championed 1, thus, $X_1 \setminus G_{21} \cup g >_2 X_2$, and agent 2 is also strictly better off. Hence, $X''$ Pareto dominates $X$. If there are no strong envy edges, we are done. So assume otherwise. We first note that the only possible strong envy edge is from 2 to 3:

\begin{itemize}

\item \emph{Agent 1 does not envy anyone}: 1 did not envy 3 in $X'$ and the bundles of 1 and 3 are the same in $X'$ and $X''$. 1 does not envy 2 anymore as he prefers his own upper-half bundle and lower-half bundle to 2's upper-half bundle and lower-half bundle respectively, i.e., $X_3 \setminus G_{13} >_1 X_1 \setminus G_{21}$ (from Table~\ref{orderingupperhalf}) and $ G_{21} \geq_1 g$ (from Table~\ref{orderinglowerhalf}).

\item \emph{Agent 3 does not envy anyone}: We use a similar argument. 3 did not envy 1 in $X'$ and the bundles of 1 and 3 are the same in $X'$ and $X''$. 3 does not envy 2 as well as he prefers his own upper-half bundle and lower-half bundle to 2's upper-half bundle and lower-half bundle respectively, namely $X_2 \setminus G_{32} >_3 X_1 \setminus G_{21}$ (from Table~\ref{orderingupperhalf}) and $G_{13} \geq_3 g $ (from Table~\ref{orderinglowerhalf}). 

\item \textit{Agent 2 does not envy 1:} Note that agent 2 has his favorite upper-half bundle and values it more than 1's upper-half bundle: $X_1 \setminus G_{21} >_2 X_3 \setminus G_{13}$ (from Table~\ref{orderingupperhalf}) and 2 also values his lower-half bundle more than 1's lower-half bundle: $g >_2 G_{21}$ (from Table~\ref{orderinglowerhalf}). 
\end{itemize}

Therefore, the only possible strong envy edge is from 2 to 3 as shown below. 

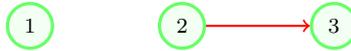
\begin{center}
\begin{tikzpicture}
[
agent/.style={circle, draw=green!60, fill=green!5, very thick},
good/.style={circle, draw=red!60, fill=red!5, very thick, minimum size=1pt},
]

%Vertices for config1
\node[agent]      (a1) at (0,0)      {$\scriptstyle{1}$};
\node[agent]      (a2) at (2,0)      {$\scriptstyle{2}$};
\node[agent]      (a3) at (4,0)     {$\scriptstyle{3}$};

%Vertices for config2
%\node[agent]      (b1) at (5,0)      {$\scriptstyle{1}$};
%\node[agent]      (b2) at (7,0)      {$\scriptstyle{2}$};
%\node[agent]      (b3) at (5,-2.5)     {$\scriptstyle{3}$};

%Edges of G_X
%\draw[->,blue,very thick] (a1)--(a3);
%\draw[->,blue,very thick] (b1)--(b3);

%Edges of M_X -Config1
%\draw[->,red,thick] (a1)--(a2);

%Edges of M_X Config2
\draw[->,red,thick] (a2) -- (a3);
%\draw[->,red,dashed,thick] (a2) --(a3);
%\draw[->,blue,dashed thick] (a3) -- (a1);

\end{tikzpicture}
\end{center}
Similar to Observation~\ref{orderingfor1}, we can now infer more about 2's preference ordering for the bundles:

\begin{observation}
\label{orderingfor2}
If 2 strongly envies 3 in $X''$, we have $X_2 \setminus G_{32} >_2 X_3 \setminus G_{13}$ and $G_{13} >_2 G_{32}$.
\end{observation}

\begin{proof}
As in Observation~\ref{orderingfor1}, we argue by contradiction. Therefore, assume that i.e. $X_2 \setminus G_{32} \leq_2 X_3 \setminus G_{13}$ or $G_{13} \leq_2 G_{32}$. If $X_2 \setminus G_{32} \leq_2 X_3 \setminus G_{13}$, then
\begin{align*}
 (X_2\setminus G_{32}) \cup G_{13}  &\leq_2 (X_3 \setminus G_{13}) \cup G_{13} \\ 
                                  &= X_3 \\
                                  &\leq_2 X_2  &(\text{since 2 did not envy 3 before})\\
                                  &<_2 (X_1 \setminus G_{21}) \cup g &(\text{as 2 is better off than before})
\end{align*} 
 implying that 2 does not envy 3, a contradiction. If $G_{13} \leq_2 G_{32}$, then 
\begin{align*}
 (X_2 \setminus G_{32}) \cup G_{13} &\leq_2 (X_2 \setminus G_{32}) \cup G_{32} \\
                                  &= X_2 \\
                                  &<_1 (X_1 \setminus G_{21}) \cup g &(\text{as 2 is better off than before})
\end{align*}
again implying that 2 does not envy 3, a contradiction.
\end{proof}

So we now have
\begin{equation}   X_3 \setminus G_{13} <_2 X_2 \setminus G_{32} <_2 X_1 \setminus G_{21}
  \quad\text{and}\quad
  G_{21} <_2 g <_2 G_{32} < G_{13}.
\end{equation}
We are ready to construct the final allocation. To this end, consider the bundle $(X_1 \setminus G_{21}) \cup G_{13}$. Note that,
\begin{align*}
( X_1 \setminus G_{21}) \cup G_{13} &>_2 (X_1 \setminus G_{21}) \cup G_{32} &(\text{as  }G_{13} >_2 G_{32} \text{ from Observation~\ref{orderingfor2}})\\ 
                                  &\geq_2 (X_1 \setminus G_{21}) \cup g   &(\text{as  }G_{32} \geq_2 g \text{ from Table~\ref{orderinglowerhalf}}) \\ 
                                  &>_2 X_2  &(\text{as 2 championed 1})
\end{align*}
Let $Z$ be a smallest cardinality subset of $(X_1 \setminus G_{21}) \cup G_{13}$ such that $Z >_2 X_2$. Since $g \not\in X_1$ and $g \not\in G_{13}$, $g \not\in Z$. We now give two allocations, depending on how much 3 values $Z$.
\begin{description}
\item[Case $Z >_3 X_3$:] Consider 
\[ X''' = \Allocationfive{X_3 \setminus G_{13}}{g}{X_2 \setminus G_{32}}{G_{32}}{Z}.\]

Since 1 was the champion of 3, we have $(X_3 \setminus G_{13}) \cup g >_1 X_1$. Thus, 1 and 3 are strictly better off, and 2 has the same bundle as in $X$. Therefore, $X'''$ Pareto dominates $X$. We still need to show that $X'''$ is EFX. 

\begin{itemize}
  \item \emph{Nobody strongly envies agent 1}: Since 1 is the champion of 3, we have that $((X_3 \setminus G_{13}) \cup g) \setminus h <_2 X_2$ and $((X_3 \setminus G_{13}) \cup g) \setminus h <_3 X_3$ for all $h \in (X_3 \setminus G_{13}) \cup g$ by Observation~\ref{nobodyenvieschampion} (part 1). As both 2 and 3 are not worse off than in $X$, neither of them strongly envies $(X_3 \setminus G_{13}) \cup g$. 
  
  \item \emph{Nobody envies agent 2}: Both 1 and 3 are strictly better off than in $X$ and they did not envy $X_2$ in $X$. Thus they do not envy $X_2$ now. 
  
  \item \emph{Nobody strongly envies agent 3}:  We first show that 1 does not envy $(X_1 \setminus G_{21}) \cup G_{13}$. This follows from the observation that 1 prefers  his own upper-half bundle to $X_1 \setminus G_{21}$ and lower-half bundle to  $G_{13}$:  $X_3 \setminus G_{13} >_1 X_1 \setminus G_{21}$ (from Table~\ref{orderingupperhalf}) and $g>_1G_{13}$ (from Table~\ref{orderinglowerhalf}). Thus $(X_3 \setminus G_{13}) \cup g >_1 (X_1 \setminus G_{21}) \cup G_{13}$. Therefore, 1 does not envy $Z$ either, as $Z \subseteq (X_1 \setminus G_{21}) \cup G_{13}$.

 Agent 2 does not strongly envy $Z$ since $Z$ is a smallest cardinality subset of $(X_1 \setminus G_{21}) \cup G_{13}$ that 2 values more than $X_2$. Thus $Z \setminus h \le_2 X_2$ for all $h \in Z$.  
\end{itemize}

\item[Case $Z \leq_3 X_3$:] Consider
\[ X''' = \Allocationfivetwo{X_3 \setminus G_{13}}{G_{32}}{Z}{X_2 \setminus G_{32}}{g}.\]

We first show that 1 is strictly better off in $X'''$ than in $X$. Observe that 
\begin{align*}
 (X_3 \setminus G_{13}) \cup G_{32} &>_1 (X_3 \setminus G_{13}) \cup G_{21} &\text{(by Observation~\ref{orderingfor1})} \\ 
                                  &\geq_1 (X_3 \setminus G_{13}) \cup g &(G_{21} \geq_1 g \text{ from Table~\ref{orderinglowerhalf})} \\
                                  &>_1 X_1                                  &\text{(as 1 championed 3)}
\end{align*} 
2 is better off as $Z >_2 X_2$ by definition of $Z$. 3 is also better off than in $X$ as it championed 2  and thus $X_2 \setminus G_{32} \cup g >_3 X_3$. Thus, all agents are strictly better off, and hence $X'''$ Pareto dominates $X$. We next show that $X'''$ is EFX.

\begin{itemize}
  \item \emph{Nobody envies agent 1}: Agent 2 does not envy 1 since 
  \begin{align*}
  (X_3 \setminus G_{13}) \cup G_{32} &<_2 (X_2 \setminus G_{32}) \cup G_{32} &\text{(by Observation~\ref{orderingfor2})}\\
                                   &= X_2\\
                                   &<_2 Z &\text{(by definition of Z).}
  \end{align*}
Agent 3 does not envy 1 either since he prefers his current upper-half bundle to and lower-half bundle to  1's upper-half bundle and lower-half bundle, respectively, i.e., $X_2 \setminus G_{32} >_3 X_3 \setminus G_{13}$ (from Table~\ref{orderingupperhalf}) and $g >_3 G_{32}$ (from Table~\ref{orderinglowerhalf}).
  
\item \emph{Nobody envies agent 2}: Observe that 1 does not envy $(X_1 \setminus G_{21}) \cup G_{13}$ since 1 is strictly better off, $G_{21} \geq_1 g >_1 G_{13}$ from Table~\ref{orderinglowerhalf}, and $G_{32} >_1 G_{21}$ by  Observation~\ref{orderingfor1}. Thus  $(X_3 \setminus G_{13}) \cup G_{32} >_1 (X_1 \setminus G_{21}) \cup G_{21} >_1 (X_1 \setminus G_{21}) \cup G_{13}$. Therefore, 1 does not envy $Z$ either as $Z \subseteq (X_1 \setminus G_{21}) \cup G_{13}$.

  Agent 3 does not envy 2 since $(X_2 \setminus G_{32}) \cup g >_3 X_3$ (see above) and $X_3 \geq_3 Z$.
  
  \item \emph{Nobody strongly envies agent 3}: Since 3 is the champion of 2, we have $((X_2 \setminus G_{32}) \cup g) \setminus h <_2 X_2$ and $((X_2 \setminus G_{32}) \cup g) \setminus h <_1 X_1$ for all $h \in (X_2 \setminus G_{32}) \cup g$ by Observation~\ref{nobodyenvieschampion} (part 1). As both 1 and 2 are strictly better off (in $X'''$) than in $X$, neither of them strongly envies $(X_2 \setminus G_{32}) \cup g$.   
\end{itemize}
\end{description}
 
We have thus shown that given an allocation $X$ such that $E_X$ has three sources and $M_X$ has a $3$-cycle, there exists an EFX allocation $Y$ Pareto dominating $X$. We summarize our main result for this section:

\begin{lemma}
\label{three-sources-mainlemma}
 Let $X$ be a partial EFX allocation and $g$ be an unallocated good. If $E_X$ has three sources, then there is an EFX allocation $Y$ Pareto dominating $X$.
\end{lemma}

\section{Existence of EFX: Two sources in $E_X$}
\label{twosources}

Let us assume that agents 1 and 2 are the sources, and let $(1,3) \in E_X$. We have two configurations for $E_X$ now, depending on whether or not $(2,3) \in E_X$. If $(2,3) \in E_X$, it is relatively straightforward to determine a new EFX allocation Pareto dominating $X$. Agent 3 is reachable from both 1 and 2 in $E_X$, and hence, if 3 champions either 1 or 2, we have a Pareto dominating EFX allocation by Observation~\ref{champion_in_subtree}. If 3 champions neither 1 nor 2, 1 and 2 must be champions of each other (Recall that no agent self-champions). Also note that 3 envies neither 1 nor 2. Therefore, by Remark~\ref{generalizing-2cycle-technique}, we have a Pareto  dominating EFX allocation.

\emph{From now on, we assume that $(2,3) \notin E_X$}. 

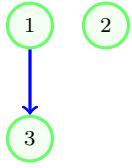
\begin{figure}
\centering
\begin{tikzpicture}
[
agent/.style={circle, draw=green!60, fill=green!5, very thick},
good/.style={circle, draw=red!60, fill=red!5, very thick, minimum size=1pt},
]

%Vertices
\node[agent]      (a1) at (0,0)      {$\scriptstyle{1}$};
\node[agent]      (a2) at (1,0)      {$\scriptstyle{2}$};
\node[agent]      (a3) at (0,-1.5)     {$\scriptstyle{3}$};

%Edges
\draw[->,blue,very thick] (a1)--(a3);
\end{tikzpicture}
\caption{Envy Graph for two sources when $(2,3) \notin E_X$: Green nodes correspond to the agents. Blue edges are the edges in $E_X$.}%\JG{We should try to make it understandable in a grayscale printout as well.}}
\label{envy-graph-2sources}
\end{figure}

The envy graph of the scenario is now as shown in Figure~\ref{envy-graph-2sources}. Next, we discuss the possible configurations of the champion graph $M_X$. We show that most configurations are easily handled. If 3 champions 1, then by Observation~\ref{champion_in_subtree}, there is a Pareto dominating EFX allocation. If 3 does not champion 1, and since 1 does not self-champion, agent 2 champions 1. If now 1 champions 2, we have a $2$-cycle in $M_X$ involving 1 and 2, and 3 envies neither of them. Therefore by Remark~\ref{generalizing-2cycle-technique}, there is a Pareto dominating EFX allocation. Thus, we may assume that 1 does not champion 2. Since 2 does not self-champion, agent 3 champions 2. \emph{There are only three possible configurations for $M_X$ now, depending on who champions 3 (only 1, only 2, both 1 and 2 as 3 does not self-champion) (see Figure~\ref{chamiongraphstates-2sources})}.

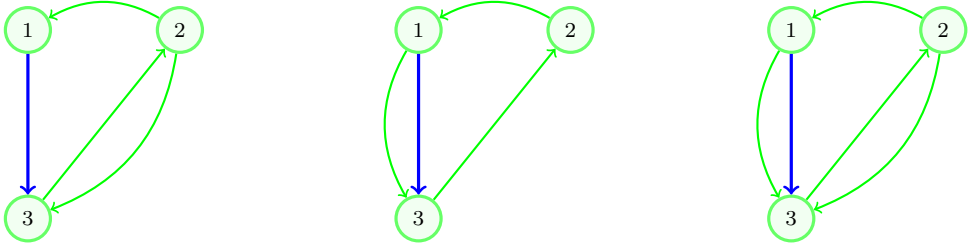
\begin{figure}[htbp]
\centering
\begin{subfigure}[b]{0.3 \textwidth}
 \begin{tikzpicture}
 [
 agent/.style={circle, draw=green!60, fill=green!5, very thick},
 good/.style={circle, draw=red!60, fill=red!5, very thick, minimum size=1pt},
 ]
 
 %Vertices for config1
 \node[agent]      (a1) at (0,0)      {$\scriptstyle{1}$};
 \node[agent]      (a2) at (2,0)      {$\scriptstyle{2}$};
 \node[agent]      (a3) at (0,-2.5)     {$\scriptstyle{3}$};
 
 %Vertices for config2
 %\node[agent]      (b1) at (5,0)      {$\scriptstyle{1}$};
 %\node[agent]      (b2) at (7,0)      {$\scriptstyle{2}$};
 %\node[agent]      (b3) at (5,-2.5)     {$\scriptstyle{3}$};

 %Edges of G_X
 \draw[->,blue,very thick] (a1)--(a3);
 %\draw[->,blue,very thick] (b1)--(b3);
 
 %Edges of M_X -Config1
 \draw[->,green,thick] (a2) edge[bend right=30] (a1);
 \draw[->,green,thick] (a2) edge[bend left=30] (a3);
 \draw[->,green,thick] (a3)--(a2);
% \draw[->,dashed,green,thick] (a1) edge[bend right=30] (a3);
 
 %Edges of M_X Config2
 %\draw[->,red,thick] (b2) edge[bend right=30] (b1);
 %\draw[->,red,thick] (b1) edge[bend right=30] (b3);
 %\draw[->,red,thick] (b3) edge[bend right=30] (b2);
 
 \end{tikzpicture}		
\end{subfigure}		
\begin{subfigure}[b]{0.3 \textwidth}
	\begin{tikzpicture}
	[
	agent/.style={circle, draw=green!60, fill=green!5, very thick},
	good/.style={circle, draw=red!60, fill=red!5, very thick, minimum size=1pt},
	]
	
	%Vertices for config1
	\node[agent]      (a1) at (0,0)      {$\scriptstyle{1}$};
	\node[agent]      (a2) at (2,0)      {$\scriptstyle{2}$};
	\node[agent]      (a3) at (0,-2.5)     {$\scriptstyle{3}$};
	
	%Vertices for config2
	%\node[agent]      (b1) at (5,0)      {$\scriptstyle{1}$};
	%\node[agent]      (b2) at (7,0)      {$\scriptstyle{2}$};
	%\node[agent]      (b3) at (5,-2.5)     {$\scriptstyle{3}$};

	%Edges of G_X
	\draw[->,blue,very thick] (a1)--(a3);
	%\draw[->,blue,very thick] (b1)--(b3);
	
	%Edges of M_X -Config1
	\draw[->,green,thick] (a2) edge[bend right=30] (a1);
%	\draw[->,dashed,green,thick] (a2) edge[bend left=30] (a3);
	\draw[->,green,thick] (a3)--(a2);
	\draw[->,green,thick] (a1) edge[bend right=30] (a3);
	
	%Edges of M_X Config2
	%\draw[->,red,thick] (b2) edge[bend right=30] (b1);
	%\draw[->,red,thick] (b1) edge[bend right=30] (b3);
	%\draw[->,red,thick] (b3) edge[bend right=30] (b2);
	
	\end{tikzpicture}		
\end{subfigure}	
\begin{subfigure}[b]{0.3 \textwidth}
	\begin{tikzpicture}
	[
	agent/.style={circle, draw=green!60, fill=green!5, very thick},
	good/.style={circle, draw=red!60, fill=red!5, very thick, minimum size=1pt},
	]
	
	%Vertices for config1
	\node[agent]      (a1) at (0,0)      {$\scriptstyle{1}$};
	\node[agent]      (a2) at (2,0)      {$\scriptstyle{2}$};
	\node[agent]      (a3) at (0,-2.5)     {$\scriptstyle{3}$};
	
	%Vertices for config2
	%\node[agent]      (b1) at (5,0)      {$\scriptstyle{1}$};
	%\node[agent]      (b2) at (7,0)      {$\scriptstyle{2}$};
	%\node[agent]      (b3) at (5,-2.5)     {$\scriptstyle{3}$};

	%Edges of G_X
	\draw[->,blue,very thick] (a1)--(a3);
	%\draw[->,blue,very thick] (b1)--(b3);
	
	%Edges of M_X -Config1
	\draw[->,green,thick] (a2) edge[bend right=30] (a1);
	\draw[->,green,thick] (a2) edge[bend left=30] (a3);
	\draw[->,green,thick] (a3)--(a2);
	\draw[->,green,thick] (a1) edge[bend right=30] (a3);
	
	%Edges of M_X Config2
	%\draw[->,red,thick] (b2) edge[bend right=30] (b1);
	%\draw[->,red,thick] (b1) edge[bend right=30] (b3);
	%\draw[->,red,thick] (b3) edge[bend right=30] (b2);
	
	\end{tikzpicture}		
\end{subfigure}	

\caption{The possible states of $M_X$ that require further discussion: Green nodes correspond to the agents. Blue edges are the edges in $E_X$ and green edges are the edges in $M_X$. There is a unique configuration of $E_X$ and three different configurations of $M_X$.}
\label{chamiongraphstates-2sources}
\end{figure}

We now show how to deal with these configurations of $M_X$. In Section~\ref{threesources}, we showed how to move from the current allocation $X$ to an allocation that Pareto dominates $X$. In Section~\ref{counter_example_monotonicity}, we show that this is impossible in this particular configuration of $E_X$ and $M_X$. More specifically, we exhibit an EFX allocation $X$ that is not Pareto dominated by any complete EFX allocation. We also show that there is no complete EFX allocation with higher Nash welfare than $X$, thereby falsifying a conjecture of Caragiannis et al.~\cite{CaragiannisGravin19}.

Recall that our potential is $\phi(X) = (v_a(X_a),v_b(X_b),v_c(X_c))$. We move to an allocation in which agent $a$ is strictly better off. We distinguish the cases: $a = 1$, $a = 2$, and $a = 3$.

Also, recall that we are in the scenario where 2 champions 1 and 2 does not envy 1. Similarly 3 champions 2 and 3 does not envy 2. Therefore, by Observation~\ref{lowerhalflessvaluable}, we have that $g \notin G_{21}$ and $g \notin G_{32}$, and hence, the bundles $X_1$ and $X_2$ decompose according to~\lref{split into upper and lower}. Also, since 2 champions 1 and 1 does not self-champion, by Observation~\ref{lowerhalfmorevaluable}, we have that $G_{21} \neq \emptyset$, and  a similar argument also shows that $G_{32} \neq \emptyset$.

\subsection{Agent $a$ is agent 1 or 3}\label{a = 1 or a = 3}

We start from the allocation

\[  X = \Allocationfive{X_1 \setminus G_{21}}{G_{21}}{X_2 \setminus G_{32}}{G_{32}}{X_3}.\]
Our goal is to determine an EFX allocation in which 1 and 3 are strictly better off ($2$ may be worse off). To this end, we consider
 
 \[  X' = \Allocationfiveone{X_3}{X_1 \setminus G_{21}}{G_{32}}{X_2 \setminus G_{32}}{g}.\]
 In $X'$, every agent is better off than in $X$: 1 is better off because $X_3 >_1 X_1$ (1 envied 3 in $E_X$). We now show that $2$ is better off: 2 championed 1 and 3 championed 2. Also, 2 did not self-champion, 2 did not envy 1  and 3 did not envy 2 . Therefore, by Observation~\ref{upperhalfvaluable}, (setting $i=k=2$, $j=1$, $i'=3$), we have that $X_1 \setminus G_{21} >_2 X_2 \setminus G_{32}$. Hence, $(X_1 \setminus G_{21}) \cup G_{32} >_2 (X_2 \setminus G_{32}) \cup G_{32} = X_2$. Thus 2 is also better off. Agent 3 is better off as 3 championed 2, and by the definition of $G_{32}$, we have $(X_2 \setminus G_{32} \cup g) >_3 X_3$. Thus $X'$ Pareto dominates $X$. If $X'$ is EFX, we are done. So assume otherwise. We show that the only possible strong envy edge will be from 1 to 2.

\begin{itemize}
\item \emph{Nobody envies 1}: Note that 1 has $X_3$ and neither 2 nor 3 envied $X_3$ earlier (3 had $X_3$ and 2 did not envy 3). Since both 2 and 3 are better off than before, they do not envy 1.

\item \emph{Nobody strongly envies 3: }{\emph{1 does not strongly envy 3 and 2 does not envy 3:} 3 championed 2 and 1 did not. Therefore, by Observation~\ref{nobodyenvieschampion} (part 1) we have $((X_2 \setminus G_{32}) \cup g) \setminus h \leq_1 X_1$ for all $h \in (X_2 \setminus G_{32}) \cup g$. Since 1 is better off than in $X$, it does not strongly envy 3.  Agent 2 does not envy 3 since its prefers both of its parts over the corresponding part of agent 3. This was argued above for the top part and follows from Observation~\ref{lowerhalfmorevaluable}}

\item \emph{3 does not envy 2}: 3 championed 2 and 3 did not envy 2 earlier. Therefore by Observation~\ref{lowerhalflessvaluable} we have that $G_{32} <_3 g$. Therefore $(X_1 \setminus G_{21}) \cup G_{32} <_3 (X_1 \setminus G_{21}) \cup g$. Since 2 championed 1 and 3 did not, by Observation~\ref{nobodyenvieschampion} (part 2), we have $((X_1 \setminus G_{21}) \cup g) \leq_3 X_3$. Since 3 is better off than in $X$, 3 does not envy 2. 
\end{itemize}

Thus, the only strong envy edge  is from 1 to 2. The current state of the envy-graph is depicted below:

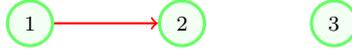
\begin{center}
\begin{tikzpicture}
[
agent/.style={circle, draw=green!60, fill=green!5, very thick},
good/.style={circle, draw=red!60, fill=red!5, very thick, minimum size=1pt},
]

%Vertices for config1
\node[agent]      (a1) at (0,0)      {$\scriptstyle{1}$};
\node[agent]      (a2) at (2,0)      {$\scriptstyle{2}$};
\node[agent]      (a3) at (4,0)     {$\scriptstyle{3}$};

%Vertices for config2
%\node[agent]      (b1) at (5,0)      {$\scriptstyle{1}$};
%\node[agent]      (b2) at (7,0)      {$\scriptstyle{2}$};
%\node[agent]      (b3) at (5,-2.5)     {$\scriptstyle{3}$};

%Edges of G_X
%\draw[->,blue,very thick] (a1)--(a3);
%\draw[->,blue,very thick] (b1)--(b3);

%Edges of M_X -Config1
\draw[->,red,thick] (a1)--(a2);

%Edges of M_X Config2
%\draw[->,red,thick] (b2) edge[bend right=30] (b1);
%\draw[->,red,thick] (b1) edge[bend right=30] (b3);
%\draw[->,red,thick] (b3) edge[bend right=30] (b2);

\end{tikzpicture}
\end{center}

Let $Z$ be a smallest cardinality subset of $(X_1 \setminus G_{21}) \cup G_{32}$ that 2 values more than $\max_2((X_2 \setminus G_{32}) \cup g, X_3)$, where $\max_2((X_2 \setminus G_{32}) \cup g, X_3)$ is defined as the more valuable bundle out of $(X_2 \setminus G_{32}) \cup g$ and $X_3$ according to 2. {Note that $\max_2((X_2 \setminus G_{32}) \cup g, X_3) \le_2 (X_1 \setminus G_{21}) \cup G_{32}$ since 2 does not envy neither 1 nor 3 in $X'$. Since the instance is non-degenerate, the inequality is strict, and hence $Z$ exists.} We now consider two allocations depending on 1's value for $Z$.

\begin{description}

\item[Case $Z \leq_1 X_3$:]  We replace 2's current bundle with $Z$ and obtain 
\[ X'' = \myBoxtwo{X_3}{1}\hspace{2em}\myBoxtwo{Z}{2}\hspace{2em}\myBox{X_2 \setminus G_{32}}{g}{3}\]
Agents 1  and 3 have the same bundles as in $X'$ and hence are strictly better off than in $X$. Thus, $X''$ dominates $X$, as $a=1$ or $a=3$ and we improve $a$ strictly. We next show that $X''$ is EFX. Since the only bundle we have changed is that of 2, and there were no strong envy edges between 1 and 3 earlier, it suffices to show that there are no strong envy edges to and from 2.

\begin{itemize}
 \item \emph{Nobody envies 2}: 3 did not envy the set $(X_1 \setminus G_{21}) \cup G_{32}$. As $Z \subseteq (X_1 \setminus G_{21}) \cup G_{32}$, agent 3  does not envy $Z$ either . 1 does not envy $Z$ because we are in the case where $Z \leq_1 X_3$. 

 \item \emph{2 does not envy anyone}: This follows from the definition of $Z$ itself since $Z >_2 \max_2((X_2 \setminus G_{32}) \cup g,X_3)$.
\end{itemize}

\item[Case $Z >_1 X_3$:] In this case, we consider

\[X'' = 
\myBoxtwo{Z}{1}\hspace{2em}\myBoxtwo{\mathit{max}_2((X_2 \setminus G_{32}) \cup g,X_3)}{2}\hspace{2em}\myBoxtwo{\mathit{min}_2((X_2 \setminus G_{32}) \cup g,X_3)}{3}  \]

Agent 1 is still strictly better off than in $X$ as we are in the case $Z >_1 X_3 >_1 X_1$, and agent 3 is not worse off than before as both $X_3$ and $(X_2 \setminus G_{32}) \cup g$ are at least as valuable to him as his previous bundle $X_3$. We first show that $X''$ is EFX.

\begin{itemize}

 \item \emph{1 does not envy anyone}: We are in the case where $Z >_1 X_3$ and 1 did not envy $(X_2 \setminus G_{32}) \cup g$ when he had $X_3$ itself (and now 1 is better off than with $X_3$). Thus, 1 does not envy anyone.
 
 \item \emph{2 does not strongly envy anyone}: Since 2 chooses the better bundle out of $X_3$ and $(X_2 \setminus G_{32}) \cup g$, 2 does not envy 3. Agent 2 does not strongly envy 1 since by the definition of $Z$, we have $Z \setminus h \leq_2 \mathit{max}_2((X_2 \setminus G_{32}) \cup g,X_3)$ for all $h \in Z$. However, note that 2 envies 1. Thus, 2 does not envy 3 and does not strongly envy 1 (but envies 1).
 
 \item \emph{3 does not strongly envy anyone}: 3 did not envy the set $(X_1 \setminus G_{21}) \cup G_{32}$, \footnote{We repeat the argument made earlier: 3 championed 2 and 3 did not envy 2 earlier. Therefore, by Observation~\ref{lowerhalflessvaluable} we have that $G_{32} <_3 g$. Hence, $(X_1 \setminus G_{21}) \cup G_{32} <_3 (X_1 \setminus G_{21}) \cup g$. Since 2 championed 1 and 3 did not, by Observation~\ref{nobodyenvieschampion} (part 2), we have $((X_1 \setminus G_{21}) \cup g) \leq_3 X_3$.} and $X_3 \le X''_3$ as we argued above. Thus, 3 will not envy $Z$ either as $Z \subseteq (X_1 \setminus G_{21}) \cup G_{32}$. We next show that 3 does not strongly envy 2, observe that $(X_2 \setminus G_{32}) \cup g >_3 X_3$. Therefore, if $\min_2((X_2 \setminus G_{32}) \cup g,X_3) = (X_2 \setminus G_{32}) \cup g$, we are done. So assume $\min_2((X_2 \setminus G_{32}) \cup g,X_3) = X_3$. Since 3 championed 2 and from Observation~\ref{nobodyenvieschampion} (part 1), we have that $((X_2 \setminus G_{32}) \cup g) \setminus h \leq_3 X_3$ for all $h \in (X_2 \setminus G_{32}) \cup g$: Thus 3 does not strongly envy 2. 
 \end{itemize}

Now if $a=1$, we are done, as $X''$ is EFX and agent 1 strictly improved. So assume $a=3$. If $\min_2((X_2 \setminus G_{32}) \cup g,X_3) = (X_2 \setminus G_{32}) \cup g$, then agent 3 is strictly better off and we are done. This leaves the case that agent 3 gets $X_3$, and hence 

\[ X'' = 
\myBoxtwo{Z}{1}\hspace{2em}\myBox{X_2 \setminus G_{32}}{g}{2}\hspace{2em}\myBoxtwo{X_3}{3}
\]

The envy graph $E_{X''}$ with respect to allocation $X''$ is a path (shown below): 
1 does not envy anyone, 2 envies 1 (not strongly) and does not envy 3, and 3 envies 2. 

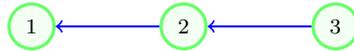
\begin{center}
\begin{tikzpicture}
[
agent/.style={circle, draw=green!60, fill=green!5, very thick},
good/.style={circle, draw=red!60, fill=red!5, very thick, minimum size=1pt},
]

%Vertices for config1
\node[agent]      (a1) at (0,0)      {$\scriptstyle{1}$};
\node[agent]      (a2) at (2,0)      {$\scriptstyle{2}$};
\node[agent]      (a3) at (4,0)     {$\scriptstyle{3}$};

%Vertices for config2
%\node[agent]      (b1) at (5,0)      {$\scriptstyle{1}$};
%\node[agent]      (b2) at (7,0)      {$\scriptstyle{2}$};
%\node[agent]      (b3) at (5,-2.5)     {$\scriptstyle{3}$};

%Edges of G_X
%\draw[->,blue,very thick] (a1)--(a3);
%\draw[->,blue,very thick] (b1)--(b3);

%Edges of M_X -Config1
\draw[->,blue,thick] (a2)--(a1);
\draw[->,blue,thick] (a3)--(a2);

%Edges of M_X Config2
%\draw[->,red,thick] (b2) edge[bend right=30] (b1);
%\draw[->,red,thick] (b1) edge[bend right=30] (b3);
%\draw[->,red,thick] (b3) edge[bend right=30] (b2);

\end{tikzpicture}
\end{center}

 Also, note that we have some unallocated goods, e.g., the goods in $G_{21}$. Recall that we argued $G_{21} \not= \emptyset$ in the paragraph just before Section~\ref{a = 1 or a = 3}. Consider any good $g' \in G_{21}$. Since $3$ is the only source in $E_{X''}$, by Corollary~\ref{singlesource}, there is an EFX allocation $X'''$ Pareto dominating $X''$, where $X'''_3 >_3 X''_3 = X_3$. Thus, we have an EFX allocation $X''' $ that dominates $X$ (as agent 3 is strictly better off and $a=3$).
\end{description}

\subsection{Agent $a$ is agent 2}
Recall that we argued just before the beginning of Section~\ref{a = 1 or a = 3} that $g \notin G_{21}$ and $g \notin G_{32}$. Thus, the current EFX allocation $X$ is 

\[ X =  \Allocationfive{X_1 \setminus G_{21}}{G_{21}}{X_2 \setminus G_{32}}{G_{32}}{X_3}\]
Our aim is to determine an EFX allocation, in which agent 2 has a bundle more valuable than $X_2$. First, observe that $(X_1 \setminus G_{21}) \cup g$ is such a bundle. As 2 championed 1, we have $(X_1 \setminus G_{21}) \cup g >_2 X_2$ by the definition of $G_{21}$. We also observe that both agents 1 and 3  value $X_3$ as least as much as $X_2$ and $(X_1 \setminus G_{21}) \cup g$.

\begin{observation}
\label{X_3 is better}
 $X_3 >_i {\max_i}(X_2,( (X_1 \setminus G_{21}) \cup g)$ for $i \in \left\{1,3\right\}$. 
\end{observation}

\begin{proof} We argue $\ge_i$; strict inequality then follows from non-degeneracy.
  
 Nobody envies 2 in $X$. Thus, $X_2 \leq_3 X_3$, and $X_2 \leq_1 X_1 <_1 X_3$ (the strict inequality holds as 1 envies 3 in $X$).

 2 is the unique champion of 1 in $X$ (both 1 and 3 do not champion 1). Therefore, by Observation~\ref{nobodyenvieschampion} (part 2), we have $(X_1 \setminus G_{21}) \cup g \leq_3 X_3$ and $(X_1 \setminus G_{21}) \cup g \leq_1 X_1 <_1 X_3$ (the strict inequality holds as 1 envies 3 in $X$). 
\end{proof}

% We now separate agents 1 and 3 depending on how much they value $X_3$ more that their respective more valuable bundle of $X_2$ and $X_1 \setminus G_{21} \cup g$.
For $i \in \left\{1,3\right\}$, let $\kappa_i$ be the size of a smallest subset $Z_i$ of $X_3$ such that $Z_i >_i \max_i((X_1 \setminus G_{21}) \cup g, X_2)$. We use the relative size of $\kappa_1$ and $\kappa_3$ to differentiate between agents 1 and 3. We use $w$ (winner) to denote the agent with the smaller value of $\kappa_i$, i.e., $w = 1$ if $\kappa_1 \le \kappa_3$ and $w = 3$ if $\kappa_1 > \kappa_3$. We use $\ell$ (loser) for the other agent. 
Consider

\[ X' = 
\myBoxtwo{X_3}{w}\hspace{2em}\myBoxtwo{\mathit{max}_{\ell}(X_2,(X_1 \setminus G_{21}) \cup g)}{\ell}\hspace{2em}\myBoxtwo{\mathit{min}_{\ell}(X_2,(X_1 \setminus G_{21}) \cup g)}{2}\]

In $X'$, the only possible strong envy edge is from $\ell$ to $w$. By Observation~\ref{X_3 is better}, $w$ envies neither $\ell$ nor 2. Note that 2  championed 1 and  therefore, $(X_1 \setminus G_{21}) \cup g >_2 X_2$, but by Observation~\ref{nobodyenvieschampion} (part 1), we have $((X_1 \setminus G_{21}) \cup g) \setminus h \leq_2 X_2$ for all $h \in (X_1 \setminus G_{21}) \cup g$. Thus, 2 gets a bundle worth at least $X_2$ and does not strongly envy $\ell$. 2  also does not envy $w$ (as he did not envy $X_3$ when he had $X_2$). $\ell$ does not envy 2 as he chooses the better bundle out of $X_2$ and $X_1 \setminus G_{21} \cup g$. Thus, the only possible strong envy edge is from $\ell$ to $w$. How we proceed then depends on whether or not $\ell$ strongly envies $w$. 

\paragraph{$\ell$ does not strongly envy $w$:}
Then $X'$ is EFX. If $\min_{\ell}(X_2, (X_1 \setminus G_{21}) \cup g) = (X_1 \setminus G_{21}) \cup g$, we are done as $X'$ dominates $X$ (2 is strictly better off and $a=2$). So assume otherwise. Then 
\[ X' = 
\myBoxtwo{X_3}{w}\hspace{2em}\myBoxtwo{X_1 \setminus G_{21} \cup g}{\ell}\hspace{2em}\myBoxtwo{X_2}{2}
\]
By Observation~\ref{X_3 is better}, $\ell$ envies $w$. Since 
 % Now we make a case distinction depending on whether $\ell$ envies $w$ or not.\KM{I have not incorporated yet that the first case cannot arise by Obs 17.}
%
% \begin{description}
% \item[Case, $\ell$ does not envy $w$:] In this case, both $\ell$ and $w$ have a bundle of at least as good as $X_3$ ($w$ has $X_3$ and $\ell$ does not envy $w$), implying that both 1 and 3 have a bundle worth at least $X_3$.  Thus 1 is strictly better off (as it envied 3 before) and 3 is not worse off. Also note that 2 is not worse off as it has $X_2$. Thus $X'$ Pareto dominates $X$.
%
% \item[Case, $\ell$ envies $w$:] 
%   We show that the envy graph $G_{X'}$ is a path with source $2$. Note that
2 only envies $\ell$, $\ell$ only envies $w$, and $w$ does not envy anyone, the envy graph $E_{X'}$ is a path with source 2. %
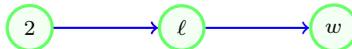
\begin{center}
\begin{tikzpicture}
[
agent/.style={circle, draw=green!60, fill=green!5, very thick},
good/.style={circle, draw=red!60, fill=red!5, very thick, minimum size=1pt},
]

%Vertices for config1
\node[agent]      (a1) at (0,0)      {$\scriptstyle{2}$};
\node[agent]      (a2) at (2,0)      {$\scriptstyle{\ell}$};
\node[agent]      (a3) at (4,0)     {$\scriptstyle{w}$};

%Vertices for config2
%\node[agent]      (b1) at (5,0)      {$\scriptstyle{1}$};
%\node[agent]      (b2) at (7,0)      {$\scriptstyle{2}$};
%\node[agent]      (b3) at (5,-2.5)     {$\scriptstyle{3}$};

%Edges of G_X
%\draw[->,blue,very thick] (a1)--(a3);
%\draw[->,blue,very thick] (b1)--(b3);

%Edges of M_X -Config1
\draw[->,red,thick] (a1)--(a2);

%Edges of M_X Config2
\draw[->,blue,thick] (a1) -- (a2);
\draw[->,blue,thick] (a2) --(a3);

\end{tikzpicture}
\end{center}
 Also, note that there are unallocated goods, namely the goods in $G_{21}$ (we argued just before the beginning of Section~\ref{a = 1 or a = 3} that $G_{21} \neq \emptyset$). Therefore, by Corollary~\ref{singlesource}, there is an EFX allocation $X''$, in which 2 is strictly better off. Thus, $X''$ dominates $X$ as 2 is strictly better off and $a = 2$. 

\paragraph{$\ell$ strongly envies $w$:} We keep removing the least valuable good \emph{according to $w$} from $w$'s bundle, until $\ell$ does not strongly envy $w$ anymore. Let $Z$ be the bundle obtained in this way. Consider 

\[ X' = 
\myBoxtwo{Z}{w}\hspace{2em}\myBoxtwo{\mathit{max}_{\ell}(X_2,(X_1 \setminus G_{21}) \cup g)}{\ell}\hspace{2em}\myBoxtwo{\mathit{min}_{\ell}(X_2,(X_1 \setminus G_{21}) \cup g)}{2}
\]

\begin{claim}
 \label{w_enviesnobody}
 $w$ does not envy 2  and $\ell$.
\end{claim}

\begin{proof}
Recall that $\kappa_w$ is the smallest cardinality of a subset of $X_3$ that $w$ still values more than $\max_{w}(X_2,(X_1 \setminus G_{21}) \cup g)$; $\kappa_w$ was defined just after Observation~\ref{X_3 is better}.  Such a set can be obtained by removing $w$'s $\lvert X_3 \rvert - \kappa_w$ least valuable goods from $X_3$. Observe that $Z$ is obtained by removing $\lvert X_3 \rvert - \lvert Z \rvert$ of $w$'s least valuable goods from  $X_3$. If $\abs{Z} \ge \kappa_w$, $w$ will envy neither $2$ nor $\ell$. If $\abs{Z} < \kappa_w \le \kappa_\ell$ (recall that $\kappa_w \le \kappa_\ell$), let $h$ be the last good removed. Then $\ell$ strongly envies $Z \cup h$ (otherwise we would not have removed $h$), meaning that there exists an $h' \in Z \cup h$ such that $(Z \cup h) \setminus h' >_{\ell} \max_{\ell}(X_2,(X_1 \setminus G_{21}) \cup g)$. Thus, there is a subset of $X_3$ of size $\lvert (Z \cup h) \setminus h'  \rvert < \kappa_w +1 -1 =\kappa_w$ that $\ell$ values more than $\max_{\ell}(X_2,(X_1 \setminus G_{21}) \cup g)$, a contradiction to $\kappa_w \le \kappa_\ell$. 
\end{proof}

The allocation $X'$ is EFX: $w$ envies neither $2$ nor $\ell$, $\ell$ does not strongly envy $w$, $\ell$ does not envy $2$, and 2 envies neither $\ell$ nor $w$. If $\min_{\ell}(X_2, (X_1 \setminus G_{21}) \cup g)$ is $X_1 \setminus G_{21} \cup g$, then we are done as $X'$  dominates $X$ (2 is strictly better off and $a= 2$). So assume otherwise. Then

\[X' = 
\myBoxtwo{Z}{w}\hspace{2em}\myBoxtwo{X_1 \setminus G_{21} \cup g}{\ell}\hspace{2em}\myBoxtwo{X_2}{2}
\]
\noindent
In $X'$, $w$ envies nobody (by Claim~\ref{w_enviesnobody}), 2 envies $\ell$, and $\ell$ may or may not envy $w$.
We distinguish cases according to whether or not $\ell$ envies $w$.
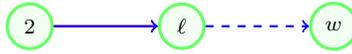
\medskip

\begin{center}
\begin{tikzpicture}
[
agent/.style={circle, draw=green!60, fill=green!5, very thick},
good/.style={circle, draw=red!60, fill=red!5, very thick, minimum size=1pt},
]

%Vertices for config1
\node[agent]      (a1) at (0,0)      {$\scriptstyle{2}$};
\node[agent]      (a2) at (2,0)      {$\scriptstyle{\ell}$};
\node[agent]      (a3) at (4,0)     {$\scriptstyle{w}$};

%Vertices for config2
%\node[agent]      (b1) at (5,0)      {$\scriptstyle{1}$};
%\node[agent]      (b2) at (7,0)      {$\scriptstyle{2}$};
%\node[agent]      (b3) at (5,-2.5)     {$\scriptstyle{3}$};

%Edges of G_X
%\draw[->,blue,very thick] (a1)--(a3);
%\draw[->,blue,very thick] (b1)--(b3);

%Edges of M_X -Config1
\draw[->,red,thick] (a1)--(a2);

%Edges of M_X Config2
\draw[->,blue,thick] (a1) -- (a2);
\draw[->,blue,dashed,thick] (a2) --(a3);

\end{tikzpicture}
\end{center}

\begin{description}
\item[Case $\ell$ envies $w$:] Then, the current envy graph is a path with 2 as the source.

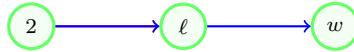
\begin{center}
\begin{tikzpicture}
[
agent/.style={circle, draw=green!60, fill=green!5, very thick},
good/.style={circle, draw=red!60, fill=red!5, very thick, minimum size=1pt},
]

%Vertices for config1
\node[agent]      (a1) at (0,0)      {$\scriptstyle{2}$};
\node[agent]      (a2) at (2,0)      {$\scriptstyle{\ell}$};
\node[agent]      (a3) at (4,0)     {$\scriptstyle{w}$};

%Vertices for config2
%\node[agent]      (b1) at (5,0)      {$\scriptstyle{1}$};
%\node[agent]      (b2) at (7,0)      {$\scriptstyle{2}$};
%\node[agent]      (b3) at (5,-2.5)     {$\scriptstyle{3}$};

%Edges of G_X
%\draw[->,blue,very thick] (a1)--(a3);
%\draw[->,blue,very thick] (b1)--(b3);

%Edges of M_X -Config1
\draw[->,red,thick] (a1)--(a2);

%Edges of M_X Config2
\draw[->,blue,thick] (a1) -- (a2);
\draw[->,blue,thick] (a2) --(a3);

\end{tikzpicture}

\end{center}

Since there are unallocated goods, namely the goods in $G_{21}$ (we argued just before the beginning of Section~\ref{a = 1 or a = 3} that $G_{21} \neq \emptyset$), by Corollary~\ref{singlesource}, there is an EFX allocation $X''$ in which agent 2 is strictly better off.  The allocation $X''$ dominates $X$ (as 2 is strictly better off and $a = 2$). 

\item[Case $\ell$ does not envy $w$:] Then the current envy graph has two sources, namely $w$ and 2, and one envy edge from 2 to $\ell$.

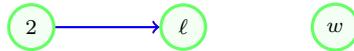
\begin{center}
\begin{tikzpicture}
[
agent/.style={circle, draw=green!60, fill=green!5, very thick},
good/.style={circle, draw=red!60, fill=red!5, very thick, minimum size=1pt},
]

%Vertices for config1
\node[agent]      (a1) at (0,0)      {$\scriptstyle{2}$};
\node[agent]      (a2) at (2,0)      {$\scriptstyle{\ell}$};
\node[agent]      (a3) at (4,0)     {$\scriptstyle{w}$};

%Vertices for config2
%\node[agent]      (b1) at (5,0)      {$\scriptstyle{1}$};
%\node[agent]      (b2) at (7,0)      {$\scriptstyle{2}$};
%\node[agent]      (b3) at (5,-2.5)     {$\scriptstyle{3}$};

%Edges of G_X
%\draw[->,blue,very thick] (a1)--(a3);
%\draw[->,blue,very thick] (b1)--(b3);

%Edges of M_X -Config1
\draw[->,red,thick] (a1)--(a2);

%Edges of M_X Config2
\draw[->,blue,thick] (a1) -- (a2);
%\draw[->,blue,dashed,thick] (a2) --(a3);

\end{tikzpicture}
\end{center}

There are at least two unallocated goods, the goods in $G_{21}$ (we argued just before the beginning of Section~\ref{a = 1 or a = 3} that $G_{21} \neq \emptyset$) and the goods in $X_3 \setminus Z$ (note that this set is not empty; we definitely have removed at least one good from $X_3$ as $\ell$ strongly envied it in $X'$). Now consider the allocation $X'$ and some $g' \in G_{21}$. If the champion of $2$ is $2$ itself or $\ell$ (definition of champion based on allocation $X'$ and the unallocated good $g'$),  by Observation~\ref{champion_in_subtree} there is an EFX allocation $Y$ where the source, namely 2, is strictly better off and hence $Y$ will dominate $X$. So assume that the champion of 2 is $w$, i.e., $w \in A_{X'}(X'_2 \cup g')$. {Let $g'' \in X_3 \setminus Z$ be the last element that we removed from $X_3$ when we constructed $Z$ from $X_3$. Then $\ell$ strongly envies $Z \cup g''$ and, according to $w$, $g''$ is the least valuable good in $Z \cup g''$.} We observe that $\ell$ is the unique champion of $w$ (definition of champion based on allocation $X'$ and the unallocated good $g''$) ,i.e., $A_{X'}(X'_w \cup g'') =  \left\{\ell\right\}$. 

\begin{observation}
 For any good $g'' \in X_3 \setminus Z$, we have $ A_{X'}(X'_w \cup g'') = \left\{\ell\right\}$.
\end{observation}

\begin{proof}
  We have $X'_w = Z$. First we show that $2 \notin A_{X'}(Z \cup g'')$. Note that $Z \cup g'' \subseteq X_3$. Since $X_2 \geq_2 X_3$ (as 2 did not envy 3 in $X$), 2 will not envy $Z \cup g''$ either.

  By the construction of $Z$, $g''$ is  $w$'s least valuable good in $Z \cup g''$. Thus, the  removal of any good from $Z \cup g''$ will result in a bundle whose value for $w$ is no more than the value of $Z$ for $w$. Therefore, $\kappa_{X'}(w,Z \cup g'') =  \lvert Z \cup g'' \rvert$\footnote{Recall that $\kappa_X(i,S)$ is the size of the smallest subset of $S$ which is more valuable to $i$ than $X_i$.}. Note that $\ell$ strongly envies $Z \cup g''$. Hence, there exists $h \in Z \cup g''$ such that $(Z \cup g'') \setminus h >_{\ell} X'_{\ell}$. Therefore, $\kappa_{X'}(\ell,Z \cup g'') \leq \lvert (Z \cup g'')\setminus h \rvert = \lvert Z \cup g'' \rvert -1 < \kappa_X(w,Z \cup g'')$. Thus, $w$ does not self-champion and hence $A_{X'}(Z \cup g'') = \left\{\ell\right\}$.
\end{proof}

Consider 
\[ X'' = 
\myBoxtwo{(X'_2 \cup g') \setminus G_{w2}}{w}\hspace{2em}\myBoxtwo{(X'_w \cup g'') \setminus G_{\ell w}}{\ell}\hspace{2em}\myBoxtwo{X'_{\ell}}{2}
\]
or equivalently
\[ X''= 
\myBoxtwo{(X_2 \cup g') \setminus G_{w2}}{w}\hspace{2em}\myBoxtwo{(Z \cup g'') \setminus G_{\ell w}}{\ell}\hspace{2em}\myBoxtwo{(X_1 \setminus G_{21}) \cup g}{2}.
\]

Note that every agent is strictly better off than in $X'$. $w$ championed 2, and by the definition of $G_{w2}$, we have $(X'_2 \cup g') \setminus G_{w2} >_w X'_w$. Similarly, $\ell$ championed $w$, and by the definition of $G_{ \ell w}$, we have $(X'_w \cup g'') \setminus G_{\ell w} >_{\ell} X'_{\ell}$. 2 is better off as 2 envied $\ell$ in $X'$ i.e. $X'_2 <_2 X'_{\ell}$. Now we have an allocation $X''$ in which agent 2 is strictly better off than it was in $X$. Thus, $X''$ dominates $X$ (as $a=$ 2).  It suffices to show that $X''$ is EFX now. To this end, observe that,

\begin{itemize}
\item \textit{Nobody strongly envies $w$:} $w$ championed 2. Thus, by Observation~\ref{nobodyenvieschampion} (part 1), we have that $((X'_2 \cup g') \setminus G_{w2}) \setminus h \leq_2 X'_2$ and $((X'_2 \cup g') \setminus G_{w2}) \setminus h \leq_{\ell} X'_{\ell}$ for all $h \in ((X'_2 \cup g') \setminus G_{w2})$. Since both 2 and $\ell$ are better off than before (in $X'$), they do not strongly envy $w$.

\item \textit{Nobody strongly envies $\ell$:} The argument is very similar to the previous case. $\ell$ championed 2. Thus, by Observation~\ref{nobodyenvieschampion} (part 1), we have that $((X'_{w} \cup g'') \setminus G_{ \ell w}) \setminus h \leq_2 X'_2$ and $((X'_{w} \cup g'') \setminus G_{\ell w}) \setminus h \leq_{w} X'_{w}$ for all $h \in ((X'_{w} \cup g'') \setminus G_{ \ell w})$. Since both 2 and $w$ are better off than before (than they were in $X'$), they do not strongly envy $w$.

\item \textit{Nobody strongly envies $2$:} Both $w$ and $\ell$ did not envy $X'_{\ell}$ ($\ell$ had $X'_{\ell}$ and $w$ did not envy $\ell$) when they had $X'_w$ and $X'_{\ell}$ itself. Both $w$ and $\ell$ are strictly better off than they were in $X'$. Therefore, they also do not envy 2.
\end{itemize} 
\end{description}

We conclude that there is an EFX allocation dominating $X$ in the case, $a = 2$ as well.

This allows us to summarize our main result for this section as follows,

\begin{lemma}
\label{two-sources-mainlemma}
 Let $X$ be a partial EFX allocation, and let $g$ be an unallocated good, where the envy graph $E_X$ has two sources. Then there is an  EFX allocation $Y$ dominating $X$. 
\end{lemma}

Having covered all the cases, we arrive at our main result:

\begin{theorem}
\label{main-thm}
 For any instance $I = \langle [3],M, \mathcal{V} \rangle$ where all $v_i \in \mathcal{V}$ are additive, an EFX allocation always exists. 
\end{theorem}

\begin{proof}
We start off with an empty allocation ($X_i = \emptyset$ for all $i \in [3]$), which is trivially EFX. As long as $X$ is not a complete EFX allocation, there is an allocation $Y$ that dominates $X$: If $E_X$ has a single source or $M_X$ has a $1$-cycle, there is a dominating EFX allocation $Y$ by Corollary~\ref{singlesource}. Lemmas~\ref{three-sources-mainlemma} and~\ref{two-sources-mainlemma} establish the existence of $Y$ when $E_X$ has multiple sources and $M_X$ does not have a $1$-cycle. Since $\phi$ is bounded from above, the process must stop. When it stops, we have arrived at a complete EFX allocation.  
\end{proof}

\section{Barriers in Current Techniques}~\label{counter_example_monotonicity}
In this section, we highlight some barriers to the current techniques for computing EFX allocations. We give an instance with three agents and seven goods such that there is a partial EFX allocation for six of the goods that is not Pareto dominated by any complete EFX allocation for the full set of goods. We also generalize this example and give an instance with a partial EFX allocation which has a Nash welfare larger than the Nash welfare of any complete EFX allocation. These examples make it unlikely that there is an iterative algorithm towards a complete EFX allocation that improves the current EFX allocation in each iteration either in the sense of Pareto domination or in the sense of Nash welfare (like the algorithms in ~\cite{TimPlaut18} and ~\cite{CKMS20}). The second example also falsifies the EFX monotonicity conjecture (see Conjecture~\ref{monotonicityconjecture}) by Caragiannis et al.~\cite{CaragiannisGravin19}.

\begin{table}[t]
\begin{center}
 \begin{tabular}{||c c c c c c c c||} 
 \hline
 & $g_1$ & $g_2$ & $g_3$ & $g_4$ & $g_5$ & $g_6$ & $g_7$ \\ [0.5ex] 
 \hline\hline
 $\mathbf{a_1}$ & $8$ & $2$ & $12$ & $2$ & $0$ & $17$ & $1$ \\ 
 \hline
 $\mathbf{a_2}$ & $5$ & $0$ & $9$ & $4$ & $10$ & $0$ &$3$ \\
 \hline
 $\mathbf{a_3}$ & $0$ & $0$ & $0$ & $0$ & $9$ & $10$ & $2$ \\
 \hline
\end{tabular}
\end{center}
\caption{An instance where no complete EFX allocation dominates the EFX allocation $X$ for the first six goods defined in the text. The valuations are assumed to be additive and the entry in row $i$ and column $j$ is the value of good $j$ for agent $i$.}
\label{Example One}
\end{table}

\begin{theorem}
  For the instance given in Table~\ref{Example One}, the partial allocation $X= \langle X_1,X_2,X_3 \rangle$, where 
\[
    X_1 = \left\{g_2,g_3,g_4 \right\} \qquad
    X_2 = \left\{g_1,g_5\right\}\qquad
    X_3 = \left\{g_6\right\},
\]
is an EFX allocation of the first six goods. No complete EFX allocation Pareto dominates $X$.
\end{theorem}
\begin{proof} Note that $v_1(X_1) = 16$, $v_2(X_2) = 15$, and $v_3(X_3) = 10$. We will show that there is no complete {EFX} allocation $X'$ with $v_1(X'_1) \ge 16$, $v_2(X'_2) \ge 15$ and $v_3(X'_3) \ge 10$. To this end, we systematically consider potential bundles $X'_1$ that can keep $a_1$'s valuation at or above 16.

Let us first assume ${g_6 \in X'_1 }$, and hence, $v_1(X'_1) \geq 17$. Now, to ensure $v_3(X'_3) \geq 10$, we need to allocate $g_5$ and $g_7$ to $a_3$. We are left with goods $g_1$, $g_2$, $g_3$ and $g_4$. In order to ensure $v_2(X'_2) \geq 15$, we definitely need to allocate $g_1$, $g_3$ and $g_4$ to $a_2$. Now even if we allocate the remaining good $g_2$ to $a_1$, we will have $v_1(X'_1) = v_1(\sset{g_2,g_6}) = 19 < 20 =v_1(\sset{g_1,g_3}) \leq {v_1}(X'_2 \setminus g_4)$. Therefore, $a_1$ will strongly envy $a_2$. Thus $g_6 \notin X'_1$.

If $g_6 \notin X'_1$ and $v_1(X'_1) \geq 16$, {$X'_1$ must contain $g_3$} (the total valuation for $a_1$ of all the goods other than $g_3$ and $g_7$ is less than 16). We need to consider several subcases. 

Assume ${g_1 \in X'_1}$ first. Since $X'_1$ already contains $g_1$ and $g_3$, the goods that can be allocated to $a_2$ and $a_3$ are $g_2$, $g_4$, $g_5$, $g_6$, and $g_7$. In order to ensure $v_2(X'_2) \geq 15$ we need to allocate $g_4$, $g_5$, and $g_7$ to $a_2$. Even if we allocate all the remaining goods ($g_2$ and $g_6$) to $a_3$, we have $v_3(X'_3) = v_3(\sset{g_3,g_6}) = 10 < 11 = v_3(\sset{g_5,g_7}) \leq v_3(X'_2 \setminus g_4)$. Therefore, $a_3$ will strongly envy $a_2$. 

Thus $g_1 \notin X'_1$. Since neither $g_1$ nor $g_6$ belongs to $X'_1$, the only way to ensure $v_1(X'_1) \geq 16$ is to at least allocate $g_2$, $g_3$, and $g_4$ to $a_1$(we can allocate more). Similarly, given that the goods not allocated yet are $g_1$, $g_5$, $g_6$, and $g_7$, the only way to ensure $v_1(X'_2) \geq 15$ is to allocate at least $g_1$ and $g_5$ to  $a_2$. Similarly, the only way to ensure $v_3(X'_3) \geq 10$ now is to allocate at least $g_6$ to $a_3$. We next show that adding $g_7$ to any one of the existing bundles will cause a violation of the EFX property.

\begin{itemize}
    \item Adding $g_7$ to $X'_1$: $a_2$ strongly envies $a_1$ as $v_2(X'_2)=15 < 16 = v_2(\left\{g_3,g_4,g_7\right\}) = v_2(X'_1 \setminus g_2)$.
    \item Adding $g_7$ to $X'_2$: $a_3$ strongly envies $a_2$ as $v_3(X'_3)=10 < 11 = v_3(\left\{g_5,g_7\right\}) = v_3(X'_2 \setminus g_1)$.
    \item Adding $g_7$ to $X'_3$: $a_1$ strongly envies $a_3$ as $v_1(X'_1)=16 < 17 = v_1(g_6) = v_1(X'_3 \setminus g_7)$.
    \end{itemize}
Thus, there exists no complete EFX allocations Pareto dominating $X$.  
\end{proof}

We now move on to the second example. We will modify the example in Table~\ref{Example One} to highlight some barriers in the existence of ``efficient" EFX allocations. There has been quite a lot of recent work aiming to compute fair allocations that are also efficient. The common measures of efficiency in economics are ``Pareto optimality" (where we cannot make any single agent strictly better off without harming another agent) and ``Nash welfare" (the geometric mean of the valuations of the agents). Quite recently, Caragiannis et al.~\cite{CaragiannisGravin19} showed that there exist partial EFX allocations that are efficient (with good guarantees on Nash welfare). In particular, they show,

\begin{theorem}[\cite{CaragiannisGravin19}]
 \label{EFXwithNSW}
  Let $X^{*} = \langle X^{*}_1,X^{*}_2, \dots, X^{*}_n \rangle$ be an allocation that maximizes the Nash welfare. Then, there exists a partial allocation $Y = \langle Y_1,Y_2,\dots, Y_n \rangle$ such that 
  \begin{itemize}
      \item For all $i \in N$ we have $Y_i \subseteq X^{*}_i$.
      \item $Y$ is EFX.
      \item $v_i(Y_i) \geq \tfrac{1}{2}v_i(X^{*}_i)$.
  \end{itemize}
\end{theorem}

In the same paper, the authors mention that if the following conjecture is true, then there exist complete EFX allocations that are efficient as well.\footnote{In their talk at EC'19 they explicitly mention this as the ``Monotonicity Conjecture".}  

\begin{conjecture}
 \label{monotonicityconjecture}
  Adding an item to an instance that admits an EFX allocation results in another instance that admits an EFX allocation with Nash welfare at least as high as that of the partial allocation before.
\end{conjecture}

We will now show that this conjecture is false, which suggests that EFX demands ``too much fairness" and some ``trade-offs with efficiency" may be necessary. In particular, we construct an instance $I'$, such that there exists a partial EFX allocation $X$ with Nash welfare $\NSW(X)$ strictly larger than the Nash welfare $\NSW(X')$ of any complete EFX allocation $X'$. From the example in Table~\ref{Example One}, it is clear that in any complete EFX allocation, we need to decrease the valuation of one of the agents. The high level idea is to modify $I$ to $I'$ such that the decrease in valuation of one of the agents is significantly more than the increase in valuation of the other agents.

\begin{table}[t]
\begin{center}
 \begin{tabular}{||c c c c c c c c||} 
 \hline
 & $g_1$ & $g_2$ & $g_3$ & $g_4$ & $g_5$ & $g_6$ & $g_7$ \\ [0.5ex] 
 \hline\hline
 $\mathbf{a_1}$ & $\varepsilon^{3}+ 6 \varepsilon^{5}$ & $2 \varepsilon^{5}$ & $10-\varepsilon^{3}$ & $\varepsilon^{3}$ & $10-2\varepsilon^{3}$ & $10 + 3\varepsilon^{5}$ & $\varepsilon^{5}$ \\ 
 \hline
 $\mathbf{a_2}$ & $\varepsilon$ & $0$ & $10-\varepsilon^{2}+\varepsilon^{6}$ & $2\varepsilon^{2}$ & $10$ & $0$ &$\varepsilon-\varepsilon^{2}$ \\
 \hline
 $\mathbf{a_3}$ & $0$ & $0$ & $0$ & $0$ & $10-\varepsilon^{4}$ & $10$ & $2\varepsilon^{4}$ \\
 \hline
\end{tabular}
\end{center}
\caption{\label{Example Two} An instance where no complete EFX allocation has larger Nash welfare than the EFX allocation $X$ for the first six goods defined in the text. The valuations are assumed to be additive and the entry in row $i$ and column $j$ is the value of good $j$ for agent $i$; $\varepsilon$ is positive, but infinitesimally small.}
\end{table}

\begin{theorem}
 \label{barrier}
 For the instance $I'$ with three agents and seven goods given in Table~\ref{Example Two}, the allocation $X = \langle X_1,X_2,X_3 \rangle$, where
\[ 
    X_1 = \left\{g_2,g_3,g_4 \right\}\qquad
    X_2 = \left\{g_1,g_5\right\}\qquad
    X_3 = \left\{g_6\right\},\]
  is an EFX allocation of the first six goods whose Nash welfare is larger than the Nash welfare of any complete EFX allocation.\footnote{The reader is encouraged to keep an eye on Table~\ref{Example Two} for the entire proof of Theorem~\ref{barrier}.}
\end{theorem}
\begin{proof}
  Observe that $\NSW(X) = ((10+2\varepsilon^{5}) \cdot (10+\varepsilon) \cdot(10))^{1/3}$. Let $X'$ be a complete EFX allocation with maximum Nash welfare. 
   
\begin{lemma}
\label{everyonehas10}
 $X'$ allocates the goods $g_3$, $g_5$ and $g_6$ to distinct agents. Additionally, 
 \begin{itemize}
     %\item $X'_1$ contains exactly one good from $\left\{g_3,g_5,g_6\right\}$.
     \item $X'_2$ contains exactly one good from $\left\{g_3,g_5\right\}$.
     \item $X'_3$ contains exactly one good from $\left\{g_5,g_6\right\}$.
 \end{itemize}
\end{lemma}

\begin{proof}
 Consider the following complete EFX allocation $\hat{X} = \langle \hat{X}_1,\hat{X}_2,\hat{X}_3 \rangle$:
 \[ \hat{X}_1 = \left\{g_6\right\}\qquad
    \hat{X}_2 = \left\{g_3,g_4,g_7\right\}\qquad
    \hat{X}_3 = \left\{g_1,g_2,g_5\right\}\]
 It is easy to verify that $\hat{X}$ is EFX and $\NSW(\hat{X}) = ((10 + 3\varepsilon^{{5}})(10+\varepsilon+\varepsilon^6)(10 - \varepsilon^4))^{{1/3}}$. Since $X'$ is a complete EFX allocation with maximum Nash welfare, we have $\NSW(X') \geq \NSW(\hat{X})$. If $g_3$, $g_5$, and $g_6$ are not allocated to distinct agents, there is an agent $a_i$ who does not get any of these goods. The valuation of this agent is at most $4 \varepsilon $ (since $\varepsilon$ is the maximum valuation of any agent for any good outside the set $\left\{g_3,g_5,g_6\right\}$). The valuation of the other two agents can be at most $3 \cdot (10 + \varepsilon) + 4\varepsilon = 30 + 7\varepsilon$ (since $\varepsilon$ is the maximum valuation of any agent for any good outside the set $\left\{g_3,g_5,g_6\right\}$, and $10 + \varepsilon$ upper bounds the maximum valuation of any good in $\left\{g_3,g_5,g_6\right\}$). Thus $\NSW(X') \leq ((4 \varepsilon) \cdot (30+7\varepsilon)^2)^{1/3} < \NSW(\hat{X})$ for sufficiently small $\varepsilon$. 
 
 A similar argument shows that $X'_2$ contains at least one good from $\left\{g_3,g_5\right\}$ and $X'_3$ contains at least one good from $\left\{g_5,g_6\right\}$ (since these are the only goods that the agents value close to $10$). Since the goods $g_3$, $g_5$, and $g_6$ are allocated to distinct agents, $a_2$ will get exactly one good from $\left\{g_3,g_5\right\}$ and $a_3$ will get exactly one good from $\left\{g_5,g_6\right\}$.   
\end{proof}

Let us denote the set $\left\{g_5,g_6,g_7\right\}$ as $\VAL$, the goods valuable for agent $a_3$. Note that $v_3(X'_3) = v_3(X'_3 \cap \VAL)$. We will now prove our claim by studying the cases that arise depending on $X'_3 \cap \VAL$. By Lemma~\ref{everyonehas10}, $X'_3 \cap \VAL$ is non-empty and contains exactly one of $g_5$ and $g_6$. Thus, $X'_3 \cap \VAL$ can be $\left\{g_5\right\}$, $\left\{g_6\right\}$, $\left\{g_5,g_7\right\}$,  or $\left\{g_6,g_7\right\}$ only. 

\begin{lemma}
\label{onlyg5}
 If $X'_3 \cap \VAL = \left\{g_5\right\}$, then $\NSW(X') <\NSW(X)$.
\end{lemma}

\begin{proof}
We have that $v_3(X'_3) = v_3(X'_3 \cap \VAL) = 10 -\varepsilon^{4}$. Lemma~\ref{everyonehas10} implies that $X'_2$ contains $g_3$ and $X'_1$ contains $g_6$. Note that $X'_1$ cannot contain any additional good other than $g_6$ as this would lead to $a_3$ strongly envying $a_1$ (note that $v_3(g_6) = 10 > 10 - \varepsilon^{4} =v_3(X'_3)$). Therefore $v_1(X'_1) = 10 + 3\varepsilon^{5}$. Now we distinguish two cases depending on whether or not $X'_2$ contains $g_1$. 
\begin{itemize}
    \item $g_1 \in X'_2$: In this case, $X'_2 = \sset{g_1,g_3}$, as otherwise $a_1$ strongly envies $a_2$ (note that $v_1(X'_1) = 10+3\varepsilon^{5} < 10+6\varepsilon^{5} = v_1(\left\{g_1,g_3\right\}$), and hence,  $v_2(X'_2) = v_2(\left\{g_1,g_3\right\}) = 10 + \varepsilon + \varepsilon^{6} -\varepsilon^{2}$. Thus,
    \[ \frac{v_1(X'_1)}{v_1(X_1)} = 1 + \frac{\varepsilon^{5}}{10+2\varepsilon^{5}}, \qquad \frac{v_2(X'_2)}{v_2(X_2)} = 1 - \frac{\varepsilon^{2}-\varepsilon^{6}}{10+\varepsilon}, \text{ and}\qquad\frac{v_3(X'_3)}{v_3(X_3)} \le 1,\]
    and hence, ${\NSW(X')}/{\NSW(X)} < 1$.   
    \item $g_1 \notin X'_2$: Then $v_2(X'_2) \leq v_2(\textup{remaining items}) = v_2(\left\{g_2,g_3,g_4,g_7\right\}) = 10 + \varepsilon + \varepsilon^{6}$, and hence,  \[ \frac{\NSW(X')}{\NSW(X)} = ((1 + \frac{\varepsilon^{5}}{10+2\varepsilon^{5}})(1 + \frac{\varepsilon^{6}}{10+\varepsilon})(1 - \frac{\varepsilon^{4}}{10}))^{{1}/{3}} < 1\]. \qedhere
\end{itemize}
\end{proof}

\begin{lemma}
\label{g5g7}
If $X'_3 \cap \VAL = \left\{g_5,g_7\right\}$, then $\NSW(X') <\NSW(X)$.
\end{lemma}

\begin{proof}
This proof follows the proof of Lemma~\ref{onlyg5} closely. We have $v_3(X'_3) = v_3(X'_3 \cap \VAL) = 10 +\varepsilon^{4}$. Lemma~\ref{everyonehas10} implies that $X'_2$ contains $g_3$ and $X'_1$ contains $g_6$. We now distinguish two cases depending on whether or not $\left\{g_1,g_4\right\} \subseteq X'_2$.
\begin{itemize}
    \item $\left\{g_1,g_4\right\} \subseteq X'_2$: Then $a_1$ strongly envies $a_2$ as $v_1(X'_1) \leq v_1(\textup{remaining items}) = v_1(\left\{g_2,g_6\right\}) = 10 +5\varepsilon^{5} < 10 + 6\varepsilon^{5} = v_1(\left\{g_1,g_3\right\}) \leq v_1(X'_2 \setminus g_4)$.
    \item $\left\{g_1,g_4\right\} \not \subseteq X'_2$. Then $v_2(X'_2) \leq v_2(\left\{g_1,g_2,g_3\right\}) =  10 + \varepsilon - \varepsilon^{2} + \varepsilon^{6}$ (not giving the less valuable $g_4$ and giving everything else that remains). Also, $v_1(X'_1) \leq v_1(\left\{g_1,g_2,g_4,g_6\right\}) = 10 + 2\varepsilon^{3} + 11\varepsilon^{5} $. Thus,
    \[ \frac{v_1(X'_1)}{v_1(X_1)} = 1 + \frac{2\varepsilon^{3} + 9\varepsilon^{5}}{10+2\varepsilon^{5}}, \qquad 
    \frac{v_2(X'_2)}{v_2(X_2)} = 1 - \frac{\varepsilon^2 - \varepsilon^6}{10 + \varepsilon},\text{ and} \qquad                 \frac{v_3(X'_3)}{v_3(X_3)} = 1 + \frac{\varepsilon^{4}}{10}                                               \],
    and hence, $\NSW(X') <\NSW(X)$. \qedhere
\end{itemize}
\end{proof}

\begin{lemma}
\label{g6g7}
If $X'_3 \cap \VAL = \left\{g_6,g_7\right\}$, then $\NSW(X') <\NSW(X)$.
\end{lemma}

\begin{proof}
We have $v_3(X'_3) = v_3(X'_3 \cap \VAL) = 10 + 2\varepsilon^{4}$. By Lemma~\ref{everyonehas10}, one of $g_3$ and $g_5$ will be allocated to each of $a_2$ and $a_1$. We argue that $g_1 \in X'_1$. If $g_1 \notin X'_1$, then
\begin{align*}
   v_1(X'_1) &\leq \mathit{max}(v_1(g_3),v_1(g_5)) + v_1(\left\{g_2,g_4\right\})\\
             &= (10-\varepsilon^{3}) + \varepsilon^{3} + 2\varepsilon^{5}\\
             &< 10 + 3\varepsilon^{5}\\
             &=v_1(g_6)\\
             &=v_1(X'_3 \setminus g_7),
\end{align*}
and hence, $a_1$ strongly envies $a_3$. 

Therefore $g_1 \in X'_1$. But we still have $v_1(X'_1) \leq \mathit{max}(v_1(g_3),v_1(g_5)) + v_1(\left\{g_1,g_2,g_4\right\}) = (10-\varepsilon^{3}) + (2\varepsilon^{3}+8\varepsilon^{5}) = 10 + \varepsilon^{3} + 8\varepsilon^{5}$. However, since $g_1 \in X'_1$, we have that $v_2(X'_2) \leq \mathit{max}(v_2(g_3),v_2(g_5)) + v_2(\left\{g_2,g_4\right\}) = 10 + 2\varepsilon^{2}$. Thus,
\[\frac{v_1(X'_1)}{v_1(X_1)} = 1 + \frac{\varepsilon^{3} + 6\varepsilon^{5}}{10+2\varepsilon^{5}},\qquad
\frac{v_2(X'_2)}{v_2(X_2)} \leq 1 - \frac{\varepsilon-2\varepsilon^{2}}{10+\varepsilon},\text{ and}\qquad
 \frac{v_3(X'_3)}{v_3(X_3)} = 1 + \frac{2\varepsilon^{4}}{10}\],
 and hence, $\NSW(X') < \NSW(X)$. 
\end{proof}

\begin{lemma}
\label{onlyg6_part1}
If $X'_3 \cap \VAL = \left\{g_6\right\}$ and $g_3 \in X'_2$, then $\NSW(X') <\NSW(X)$.
\end{lemma}

\begin{proof}
 We have $v_3(X'_3) = v_3(X'_3 \cap \VAL) = 10$. Since $g_3$ and $g_5$ are allocated to $a_1$ and $a_2$, respectively, and $g_3 \in X'_2$, we have $g_5 \in X'_1$ by Lemma ~\ref{everyonehas10}. We now distinguish two cases depending, on whether or not $g_1 \in X'_2$.
  \begin{itemize}
     \item $g_1 \in X'_2$: Then $X'_2$ cannot contain any other goods than $g_1$ and $g_3$, else $a_1$ will strongly envy $a_2$: $v_1(X'_1) \leq v_1(\textup{remaining items}) \leq v_1(\left\{g_2,g_4,g_5,g_7\right\}) = 10 - \varepsilon^{3} + 3\varepsilon^{5} < 10 + 6\varepsilon^{5} = v_1(\left\{g_1,g_3\right\})$. Therefore $v_2(X'_2) = v_2(\left\{g_1,g_3\right\}) = 10 + \varepsilon - \varepsilon^{2} + \varepsilon^{6}$. Also, note that $v_1(X'_1) \leq v_1(\left\{g_2,g_4,g_5,g_7\right\}) = 10 - \varepsilon^{3} + 3\varepsilon^{5}$. In that case, the valuations of both $a_1$ and $a_2$ decrease, and that of $a_3$ does not increase. Thus $\NSW(X') < \NSW(X)$.
    \item $g_1 \notin X'_2$: Then $X'_2$ cannot contain both of $g_4$ and $g_7$, else $a_1$ will strongly envy $a_2$: $v_1(X'_1) \leq v_1(\textup{remaining goods}) = v_1(\left\{g_1,g_2,g_5\right\}) = 10 - \varepsilon^{3} + 8\varepsilon^{5} < 10 = v_1(\left\{g_3,g_4\right\}) = v_1(X'_2 \setminus g_7)$. Therefore, $v_2(X'_2) \leq \mathit{max}(v_2(g_4), v_2(g_7)) + v_2(\textup{remaining items}) \leq \mathit{max}(v_2(g_4),$ $v_2(g_7)) + v_2(\left\{g_2,g_3\right\}) = 10 + \varepsilon - 2\varepsilon^{2} + \varepsilon^{6}$ and $v_1(X'_1) \leq v_1(\left\{g_1,g_2,g_4,g_5,g_7\right\}) = 10 + 9\varepsilon^{5}$.
    Thus,
    \[ \frac{v_1(X'_1)}{v_1(X_1)} = 1 + \frac{7\varepsilon^{5}}{10+2\varepsilon^{5}}, \qquad \frac{v_2(X'_2)}{v_2(X_2)} \leq 1 - \frac{2\varepsilon^{2} - \varepsilon^{6}}{10+\varepsilon}, \text{ and} \qquad \frac{v_3(X'_3)}{v_3(X_3)} = 1\],
    and hence, $\NSW(X') < \NSW(X)$.\qedhere
\end{itemize}
\end{proof}

\begin{lemma}
\label{onlyg6_part2}
If $X'_3 \cap \VAL = \left\{g_6\right\}$ and $g_3 \notin X'_2$, then $\NSW(X') <\NSW(X)$.
\end{lemma}

\begin{proof}
 We have $v_3(X'_3) = v_3(X'_3 \cap \VAL) = 10$. Since $g_3 \notin X'_2$, we have $g_5 \in X'_2$ and $g_3 \in X'_1$ by Lemma ~\ref{everyonehas10}. We now distinguish two cases depending on whether or not $g_7 \in X'_2$.
  \begin{itemize}
     \item $g_7 \in X'_2$: Then $X'_2$ cannot contain any other goods than $g_5$ and $g_7$, else $a_3$ will strongly envy $a_2$: $v_3(X'_3) =10 < 10 + \varepsilon^{4} = v_3(\left\{g_5,g_7\right\})$. Therefore, $v_2(X'_2) = v_2(\left\{g_5,g_7\right\}) = 10 + \varepsilon - \varepsilon^{2}$ and $v_1(X'_1) \leq v_1(\textup{remaining items}) =  v_1(\left\{g_1,g_2,g_3,g_4\right\}) = 10 + \varepsilon^{3} +  8\varepsilon^{5}$. Thus,
     \[\frac{v_1(X'_1)}{v_1(X_1)} = 1 + \frac{\varepsilon^{3} + 6\varepsilon^{5}}{10+2\varepsilon^{5}},\qquad 
     \frac{v_2(X'_2)}{v_2(X_2)} \leq 1 - \frac{\varepsilon^{2}}{10+\varepsilon},\text{ and}\qquad
     \frac{v_3(X'_3)}{v_3(X_3)} = 1\],
    and hence, $\NSW(X') < \NSW(X)$.
    \item $g_7 \notin X'_2$: Then $X'_2$ cannot contain both of $g_1$ and $g_4$ else $a_1$ will strongly envy $a_2$: $v_1(X'_1) \leq v_1(\textup{remaining goods}) = v_1(\left\{g_2,g_3,g_7\right\}) = 10 - \varepsilon^{3} + 3\varepsilon^{5} < 10 -\varepsilon^3 + 6\varepsilon^{5} = v_1(\left\{g_1,g_5\right\}) = v_1(X'_2 \setminus g_4)$. Now we consider two cases depending on whether or not $g_1 \in X'_2$.
     \begin{itemize}
         \item $g_1 \in X'_2$: Then $X'_2$ cannot have $g_4$. Thus $v_2(X'_2) \leq v_2(g_1) + v_2(\textup{remaining items}) = v_2(g_1) + v_2(\left\{g_2,g_5\right\}) = 10 + \varepsilon = v_2(X_2)$. Note that $X'_1$ cannot have all of the remaining goods $g_2,g_3,g_4,g_7$, else $a_2$ will strongly envy $a_1$: $v_2(X'_2) \leq  10 + \varepsilon < 10+ \varepsilon + \varepsilon^{6} = (10-\varepsilon^{2}+\varepsilon^{6}) + (2\varepsilon^{2}) + (\varepsilon-\varepsilon^{2}) = v_2(\left\{g_3,g_4,g_7\right\}) = v_2(\left\{g_2,g_3,g_4,g_7\right\} \setminus g_2)$. Therefore, $X'_1$ is a strict subset of  $\left\{g_2,g_3,g_4,g_7\right\}$, and it should contain $g_7$ (as we are in the case where neither $X'_2$ nor $X'_3$ can have $g_7$). Since $a_1$'s valuation for $g_7$ is strictly less than his valuation for any of $g_2$, $g_3$, and $g_4$, we have that $v_1(X'_1) < v_1(\left\{g_2,g_3,g_4\right\}) = v_1(X_1)$. Since we are in the case where $v_2(X'_2) \leq v_2(X_2)$ and $v_3(X'_3) = v_3(X_3)$, we have $\NSW(X') < \NSW(X)$.
         
         \item $g_1 \notin X'_2$: Then $v_2(X'_2) \leq v_2(\textup{remaining items}) = v_2(\left\{g_2,g_4,g_5\right\}) = 10 + 2\varepsilon^{2}$ and $v_1(X'_1) \leq v_1(\left\{g_1,g_2,g_3,g_4,g_7\right\}) = 10 + \varepsilon^{3} + 9\varepsilon^{5}$. Thus,
         \[\frac{v_1(X'_1)}{v_1(X_1)} = 1 + \frac{\varepsilon^{3}+ 7\varepsilon^{5}}{10+2\varepsilon^{5}},\qquad  \frac{v_2(X'_2)}{v_2(X_2)} \leq 1 - \frac{\varepsilon-2\varepsilon^{2}}{10+\varepsilon},\text{ and}\qquad \frac{v_3(X'_3)}{v_3(X_3)} = 1\],
         and hence, $\NSW(X') < \NSW(X)$. \qedhere 
     \end{itemize}
\end{itemize}
\end{proof}

Lemmas~\ref{onlyg6_part1} and~\ref{onlyg6_part2} immediately imply the following:

\begin{lemma}
 \label{onlyg6}
 If $X'_3 \cap \VAL = \left\{g_6\right\}$, then $\NSW(X') <\NSW(X)$.
\end{lemma}
We are now ready to complete the proof. 
 Lemma~\ref{everyonehas10} implies that $a_3$ gets exactly one good from $\left\{g_5,g_6\right\}$. Thus, $X'_3 \cap \VAL \neq \emptyset$, and $ \left\{g_5,g_6\right\} \not \subseteq X'_3 \cap \VAL $. So $X'_3 \cap \VAL \in \left\{\left\{g_5\right\},\left\{g_6\right\},\left\{g_5,g_7\right\},\left\{g_6,g_7\right\}\right\}$. However, Lemmas~\ref{onlyg5},~\ref{g5g7},~\ref{g6g7}, and~\ref{onlyg6} imply that in all of these cases, $\NSW(X') < \NSW(X)$.
\end{proof}

\section{Conclusion}
In this paper, we have shown that EFX allocations always exist when we have three agents with additive valuations. Our proof is constructive and leads to a pseudo-polynomial algorithm.  We have identified some crucial barriers in the current techniques and have overcome them with novel techniques. We feel that this is step towards resolving the bigger question whether EFX allocations always exist when we have $n$ agents. 

Our proofs crucially use additivity and do not work for more general valuation functions like submodular or subadditive. Therefore, an ideal next step would be to investigate EFX allocations with three agents, but more general valuations.

We also showed some barriers to finding \emph{efficient} EFX allocations (EFX allocations with high Nash social welfare). While {efficient} approximate EFX allocations or {efficient} EFX allocations with bounded charity exist, it is unclear how much efficiency we can guarantee for complete EFX allocations---i.e., what trade-off with efficiency is required to guarantee fairness.

\section*{Acknowledgements}
We would like to thank Hannaneh Akrami, Corinna Coupette, Kavitha Telikepalli and Alkmini Sgouritsa for helpful discussions. We thank Corinna Coupette also for a careful reading of the manuscript. This work is partially supported by NSF Grants CCF-1755619 (CRII) and CCF-1942321 (CAREER).

%\bibliographystyle{alpha}
%\bibliography{EFX}

\newcommand{\etalchar}[1]{$^{#1}$}

\end{document}